\documentclass[11pt]{article}

\usepackage[margin=1.25in]{geometry}

\usepackage{amsmath}
\usepackage{lecnotes}

\newcommand{\ddd}{\raisebox{0.2em}[1.1em]{$\vdots$}}

\newcommand{\DD}{\mathcal{D}}

\newcommand{\ms}[1]{\mathsf{#1}}

\newcommand{\mb}[1]{\mathbf{#1}}

\newcommand{\uscore}{\mbox{\tt\char`\_}}
\newcommand{\arrow}{\mathbin{\rightarrow}}

\newcommand{\ctxrm}{\mathrel{\backslash}}

\newcommand{\tensor}{\otimes}
\newcommand{\lolli}{\multimap}

\newcommand{\with}{\mathbin{\binampersand}}
\newcommand{\one}{\mb{1}}
\newcommand{\zero}{\mb{0}}

\newcommand{\semi}{\mathrel{;}}
\newcommand{\jtrue}{\mathit{true}}

\newcommand{\mmode}[1]{{\mathchoice{\ms{#1}}{\ms{#1}}{\scriptscriptstyle\ms{#1}}{\scriptscriptstyle\ms{#1}}}}
\newcommand{\mL}{\mmode{L}}
\newcommand{\mU}{\mmode{U}}
\newcommand{\mS}{\mmode{S}}
\newcommand{\mX}{\mmode{X}}
\newcommand{\mV}{\mmode{V}}
\newcommand{\up}{{\uparrow}}
\newcommand{\down}{{\downarrow}}

\newcommand{\rgt}{\rhd}
\newcommand{\lft}{\blacktriangleleft}
\newcommand{\pss}{\blacktriangleright}

\newcommand{\plus}{\oplus}
\newcommand{\syn}{\Longrightarrow}
\newcommand{\ssyn}{\syn}
\newcommand{\chk}{\Longleftarrow}
\newcommand{\join}{\mathbin{;}}
\newcommand{\lub}{\sqcup}
\newcommand{\uses}{\mathrel{/}}
\newcommand{\bbar}{\mathord{\Vert}}

\title{Adjoint Natural Deduction\\(Extended Version)}
\author{\makebox[12em]{\begin{tabular}{c}Junyoung Jang \\ McGill University\end{tabular}}
  \and 
  \makebox[12em]{\begin{tabular}{c}Sophia Roshal \\ Carnegie Mellon University\end{tabular}}
  \and
  \makebox[12em]{\begin{tabular}{c}Frank Pfenning \\ Carnegie Mellon University\end{tabular}}
  \and
  \makebox[12em]{\begin{tabular}{c}Brigitte Pientka \\ McGill University\end{tabular}}}
\date{February 2, 2024}

\newcommand{\runningtitle}{Adjoint Natural Deduction}
\newcommand{\runningdate}{February 2, 2024}

\begin{document}

\maketitle

\begin{abstract}
  Adjoint logic is a general approach to combining multiple logics with
  different structural properties, including linear, affine, strict, and
  (ordinary) intuitionistic logics, where each proposition has an intrinsic mode
  of truth.  It has been defined in the form of a sequent calculus because the
  central concept of independence is most clearly understood in this form, and
  because it permits a proof of cut elimination following standard techniques.

  In this paper we present a natural deduction formulation of adjoint logic and
  show how it is related to the sequent calculus.  As a consequence, every
  provable proposition has a verification (sometimes called a long normal form).
  We also give a computational interpretation of adjoint logic in the form of a
  functional language and prove properties of computations that derive from the
  structure of modes, including freedom from garbage (for modes without
  weakening and contraction), strictness (for modes disallowing weakening), and
  erasure (based on a preorder between modes).  Finally, we present a
  surprisingly subtle algorithm for type checking.
\end{abstract}

\section{Introduction}
\label{sec:intro}

A \emph{substructural logic} provides fine control over the use of assumptions
during reasoning.  It usually does so by denying the general sequent calculus
rules of \emph{contraction} (which permits an antecedent to be used more than
once) and \emph{weakening} (which permits an antecedent not to be used).
Instead, these rules become available only for antecedents of the form ${!}A$.
Ever since the inception of linear logic \citep{Girard87tcs}, researchers have
found applications in programming languages, for example, to avoid garbage
collection \citep{Girard87tapsoft}, soundness of imperative update
\citep{Wadler90ifip}, the chemical abstract machine \citep{Abramsky93},
and session-typed communication \citep{Caires10concur,Wadler12icfp}, to name just a few.

Besides linear logic, there are other substructural logics and type systems of
interest.  For example, \emph{affine logic} denies general contraction but
allows weakening and is the basis for the type system of Alms \citep{Tov11popl}
(an affine functional language) and \citet{Rust} (an imperative
language aimed at systems programming).

If we deny weakening but accept contraction we obtain \emph{strict logic} (a
variant of \emph{relevance logic}) where every assumption must be used at least
once.  On the programming language side, this corresponds to \emph{strictness},
which allows optimizations in otherwise nonstrict functional languages such as
Haskell \citep{Mycroft80sop}.  Interestingly, Church's original $\lambda{I}$
calculus \citep{Church41book} was also strict in this sense.

The question arises how we can combine such features, both in logics and in type
systems.  Recently, this question has been tackled through graded or
quantitative type systems (see for example,
\citet{Atkey18lics,Moon:ESOP21,Choudhury:POPL21,Wood22esop,Abel:ICFP23}). The
essential idea is to track and reason explicitly about the usage of a given
assumption through grades. This provides very fine-grained control and allows us
to, for example, model linear, strict, and unrestricted usage of assumptions
through graded modalities.  In this paper, we pursue an alternative taking a
proof-theoretic view with the goal of building a computational interpretation.
There are three possible options that emerge from existing proof-theoretic
explorations that could serve as a foundation of such a computational
interpretation.  The first one is by \emph{embedding}.  For example, we can
embed (structural) intuitionistic logic in linear logic writing ${!}A \lolli B$
for $A \arrow B$.  Similarly, we can embed affine logic in linear logic by
mapping hypotheses $A$ to $A \with \one$ so they do not need to be used.  The
difficulties with such embeddings is that, often, they neither respect proof
search properties such as focusing \citep{Andreoli92jlc} nor do they achieve a
desired computational interpretation.

A second approach is taken by \emph{subexponential linear logic}
\citep{Nigam09ppdp,Nigam16jlc,Kanovich18ijcar} that defines multiple
subexponential modalities ${!}^mA$, where each \emph{mode} $m$ has a specific
set of structural properties.  As in linear logic, all inferences are carried
out on linear formulas, so while it resolves some of the issues with
embeddings, it still requires frequent movement into the linear layer
using explicit subexponentials.

We pursue a third approach, pioneered by \citet{Benton94csl} who symmetrically
combined (structural) intuitionistic logic with (purely) linear intuitionistic
logic.  He employs two adjoint modalities that switch between the two layers and
works out the proof theoretic and categorical semantics.  This approach has the
advantage that one can natively reason and compute within the individual logics,
so we preserve not only provability but the fine structure of proofs and proof
reduction from each component.  This has been generalized in prior work
\citep{Reed09un,Pruiksma18un} by incorporating from subexponential linear logic
the idea to have a preorder between modes $m \geq k$ that must be compatible
with the structural properties of $m$ and $k$ (explained in more detail in
\autoref{sec:seq}).  This means we can now also model intuitionistic S4
\citep{Pfenning01mscs} and lax logic \citep{Benton98}, representing comonadic
and monadic programming, respectively. We hence arrive at a unifying calculus firmly rooted in proof theory
that is more general than previous graded modal type systems in that we can
\emph{construct} monads as well as comonads. We will briefly address dependently typed variations
of the adjoint approach in \autoref{sec:conclusion}.

Most substructural logics and many substructural type systems are most clearly
formulated as sequent calculi.  However, natural deduction has not only an
important foundational role \citep{Gentzen35,Prawitz65,Dummett91}, it also has
provided a simple and elegant notation for functional programs through the
Curry-Howard correspondence \citep{Howard69}.  We therefore develop a system of
natural deduction for adjoint logic that, in a strong sense, corresponds to the
original sequent formulation.  It turns out to be surprisingly subtle because we
have to manage not only the substructural properties that may be permitted or
not, but also respect the preorder between modes.  We show that the our calculus
satisfies some expected properties like \emph{substitution} and has a natural
notion of \emph{verification} that corresponds to proofs in long normal form,
satisfying a subformula property.

In order to illustrate computational properties, we also give an abstract
machine and show the consequences of the mode structure: freedom from garbage
for linear modes (that is, modes admitting neither weakening nor contraction),
strictness for modes that do not admit weakening, and erasure for modes that a
final value may not depend on, based on the preorder of modes.  We close with an
algorithmic type checker for our language which, again, is surprisingly subtle.

\section{Adjoint Sequent Calculus}
\label{sec:seq}

We briefly review the adjoint sequent calculus from \citet{Pruiksma18un}.  We
start with a standard set of possibly substructural propositions, indexing each with a
\emph{mode of truth}, denoted by $m, k, n, r$.  Propositions are perhaps best
understood by using their linear meaning as a guide, so we uniformly use the
notation of linear logic.  Also, for programming convenience, we generalize the
usual binary and nullary disjunction ($A \plus B$ and $\zero$) and conjunction
($A \with B$ and $\top$) by using labeled disjunction ${\plus}\{\ell : A_m^\ell\}_{\ell \in L}$
and conjunction ${\with}\{\ell : A_m^\ell\}_{\ell \in L}$.  From the linear logical
perspective, these are \emph{internal} and \emph{external} choice, respectively;
from the programming perspective they are \emph{sums} and \emph{products}.
We write $P_m$ for atomic propositions of mode $m$.
\[
  \begin{array}{llcll}
    \mbox{Propositions} & A_m, B_m & ::=
    & P_m \mid A_m \lolli B_m \mid {\with}\{\ell : A_m^\ell\}_{\ell \in L}
      \mid \up^m_k A_k & \mbox{(negative)} \\
                        & &
    & \null \mid A_m \tensor B_m \mid \one_m \mid {\plus}\{\ell : A_m^\ell\}_{\ell \in L}
      \mid \down^n_m A_n & \mbox{(positive)} \\[1ex]
    \mbox{Contexts} & \Gamma & ::= & \cdot \mid \Gamma, x:A_m & \mbox{(unordered)}
  \end{array}
\]
Each mode $m$ comes with a set $\sigma(m) \subseteq \{\mathsf{W},\mathsf{C}\}$ of \emph{structural
  properties}, where $\mathsf{W}$ stands for weakening and $\mathsf{C}$ stands for contraction.
We further have a preorder $m \geq r$ that specifies that a proof of the
succedent $C_r$ may depend on an antecedent $A_m$.  This is enforced using the
presupposition that in a sequent $\Gamma \vdash C_r$, every antecedent $A_m$ in
$\Gamma$ must satisfy $m \geq r$, written as $\Gamma \geq r$.  We have the
additional stipulation of \emph{monotonicity}, namely that $m \geq k$ implies
$\sigma(m) \supseteq \sigma(k)$.  This is required for cut elimination to hold.
Furthermore, we presuppose that in $\up^m_k A_k$ we have $m \geq k$ and for
$\down^n_m A_n$ we have $n \geq m$.  Also, contexts may not have any repeated
variables and we will implicitly apply variable renaming to maintain this
presupposition.  Finally, we abbreviate $\cdot, x : A$ as just $x : A$.

In preparation for natural deduction, instead of explicit rules of weakening and
contraction (see \citep{Pruiksma18un} for such a system) we have a context merge
operation $\Gamma_1 \join \Gamma_2$.  Since, as usual in the sequent calculus,
we read the rules bottom-up, it actually describes a nondeterministic split of
the context that is pervasive in the presentations of linear logic \citep{Andreoli92jlc}.
\[
  \begin{array}{rclcll}
    (\Gamma_1, x:A_m) & \join & (\Gamma_2, x:A_m) & = & (\Gamma_1 \join \Gamma_2), x : A_m
    & \mbox{provided $\mathsf{C} \in \sigma(m)$} \\
    (\Gamma_1, x:A_m) & \join & \Gamma_2 & = & (\Gamma_1 \join \Gamma_2), x : A_m 
    & \mbox{provided $x \not\in \rm{dom}(\Gamma_2)$} \\
    \Gamma_1 & \join & (\Gamma_2, x:A_m) & = & (\Gamma_1 \join \Gamma_2), x : A_m
    & \mbox{provided $x \not\in \rm{dom}(\Gamma_1)$} \\
    (\cdot) & \join & \Gamma_2 & = & \Gamma_2 \\
    \Gamma_1 & \join & (\cdot) & = & \Gamma_1
  \end{array}
\]
Note that the context merge is a partial operation, which prevents, for example,
the use of an antecedent without contraction in both premises of the
${\tensor}R$ rule.

The complete set of rules can be found in \autoref{fig:seqi}.
In the rules, we write $\Gamma_{\mathsf{W}}$ for a context in which weakening can be applied
to every antecedent, that is, $\mathsf{W} \in \sigma(m)$ for every antecedent $x : A_m$.
Also, as is often the case in presentations of the sequent calculus, we omit explicit
variable names that tag antecedents. We only discuss the rules for $\down^n_m A_n$
because they illustrate the combined reasoning about structural properties and modes.

\begin{figure}[ht!]
  \begin{rules}
    \infer[\ms{id}]
    {\Gamma_\mathsf{W} \join A_m \vdash A_m}
    {}
    \hspace{3em}
    \infer[\ms{cut}]
    {\Gamma \join \Gamma' \vdash C_r}
    {\Gamma \geq m \geq r
      & \Gamma \vdash A_m
      & \Gamma', A_m \vdash C_r}
    \\[1em]\hline\\[1ex]
    \infer[{\lolli}R]
    {\Gamma \vdash A_m \lolli B_m}
    {\Gamma, A_m \vdash B_m}
    \\[1em]
    \infer[{\lolli}L]
    {\Gamma \join \Gamma' \join A_m \lolli B_m \vdash C_r}
    {\Gamma \geq m
      & \Gamma \vdash A_m
      & \Gamma', B_m \vdash C_r}
    \\[1em]
    \infer[{\with}R]
    {\Gamma \vdash {\with}\{\ell : A_m^\ell\}_{\ell \in L}}
    {\Gamma \vdash A_m^\ell \quad (\forall \ell \in L)}
    \hspace{3em}
    \infer[{\with}L]
    {\Gamma \join {\with}\{\ell : A_m^\ell\}_{\ell \in L} \vdash C_r}
    {\Gamma, A_m^\ell \vdash C_r\quad (\ell \in L)}
    \\[1em]
    \infer[{\tensor}R]
    {\Gamma \join \Gamma' \vdash A_m \tensor B_m}
    {\Gamma \vdash A_m
      & \Gamma' \vdash B_m}
    \hspace{3em}
    \infer[{\tensor}L]
    {\Gamma \join A_m \tensor B_m \vdash C_r}
    {\Gamma, A_m, B_m \vdash C_r}
    \\[1em]
    \infer[{\one}R]
    {\Gamma_\mathsf{W} \vdash \one_m}
    {}
    \hspace{3em}
    \infer[{\one}L]
    {\Gamma \join \one_m \vdash C_r}
    {\Gamma \vdash C_r}
    \\[1em]
    \infer[{\plus}R]
    {\Gamma \vdash {\plus}\{\ell : A_m^\ell\}_{\ell \in L}}
    {\Gamma \vdash A^\ell_m \quad (\ell \in L)}
    \hspace{3em}
    \infer[{\plus}L]
    {\Gamma \join {\plus}\{\ell : A_m^\ell\}_{\ell \in L} \vdash C_r}
    {\Gamma, A_m^\ell \vdash C_r \quad (\forall \ell \in L)}
    \\[1em]\hline\\[1ex]
    \infer[{\up}R]
    {\Gamma \vdash \up_k^m A_k}
    {\Gamma \vdash A_k}
    \hspace{3em}
    \infer[{\up}L]
    {\Gamma \join \up_k^m A_k \vdash C_r}
    {k \geq r &
      \Gamma, A_k \vdash C_r}
    \\[1em]
    \infer[{\down}R]
    {\Gamma_{\mathsf{W}} \join \Gamma' \vdash \down^n_m A_n}
    {\Gamma' \geq n & \Gamma' \vdash A_n}
    \hspace{3em}
    \infer[{\down}L]
    {\Gamma \join \down^n_m A_n \vdash C_r}
    {\Gamma, A_n \vdash C_r}
  \end{rules}
  \caption{Implicit Adjoint Sequent Calculus}
  \label{fig:seqi}
\end{figure}

First, the ${\down}R$ rule.
\begin{rules}
  \infer[{\down}R]
  {\Gamma_{\mathsf{W}} \join \Gamma' \vdash \down^n_m A_n}
  {\Gamma' \geq n & \Gamma' \vdash A_n}
\end{rules}
Because we presuppose the conclusion is well-formed, we know
$\Gamma_{\mathsf{W}} \semi \Gamma' \geq m$ since $\down^n_m A_n$ has mode $m$.  Again, by
presupposition $n \geq m$ and we have to explicitly check that $\Gamma' \geq n$
because it doesn't follow from knowing that $\Gamma_{\mathsf{W}} \semi \Gamma' \geq m$.
There may be some antecedents $A_k$ in the conclusion such that $k \not\geq n$.
If the mode $k$ admits weakening, we can sort them into $\Gamma_{\mathsf{W}}$.  If it does
not, then the rule is simply not applicable.

On to the ${\down}L$ rule:
\begin{rules}
  \infer[{\down}L]
  {\Gamma \join \down^n_m A_n \vdash C_r}
  {\Gamma, A_n \vdash C_r}
\end{rules}
By presupposition on the conclusion, we know $\Gamma \semi \down^n_m A_n \geq r$
which means that $\Gamma \geq r$ and $m \geq r$.  Since $n \geq m$ we have
$n \geq r$ by transitivity, so $\Gamma, A_n \geq r$ and we do not need any
explicit check.  The formulation of the antecedents in the conclusion
$\Gamma \semi \down^n_m A_n$ means that if mode $m$ admits contraction, then the
antecedent $\down^n_m A_n$ may also occur in $\Gamma$, that is, it may be
preserved by the rule.  If $m$ does not admit contraction, this occurrence of
$\down^n_m A_n$ is not carried over to the premise.

This implicit sequent calculus satisfies the expected theorems, due to
\citet{Reed09un,Pruiksma18un} and, most closely reflecting the precise form of
our formulation, \citet{Pruiksma24phd}.  They follow standard patterns,
modulated by the substructural properties and the preorder on modes.

\begin{theorem}[Admissibility of Weakening and Contraction]
  The following are admissible:
  \begin{rules}
    \infer-[\ms{weaken}]
    {\Gamma_{\mathsf{W}} \join \Gamma \vdash A_m}
    {\Gamma_{\mathsf{W}} \geq m & \Gamma \vdash A_m}
    \hspace{3em}
    \infer-[\ms{contract}]
    {\Gamma, A_m \vdash C_r}
    {\mathsf{C} \in \sigma(m) & \Gamma, A_m, A_m \vdash C_r}
  \end{rules}
\end{theorem}

\begin{theorem}[Admissibility of Cut and Identity]
  \label{thm:cut-id}
  \mbox{}
  \begin{enumerate}[(i)]
  \item In the system without cut, cut is admissible.
  \item In the system with identity restricted to atoms $P_m$, the
    general identity is admissible.
  \end{enumerate}
\end{theorem}

We call a proof \emph{cut-free} if it does not contain cut and \emph{long} if
the identity is restricted to atomic propositions $P$.  It is an immediate
consequence of \autoref{thm:cut-id} that every derivable sequent has a long
cut-free proof.  The subformula property of cut-free proofs directly implies
that a cut-free proof of a sequent $\Gamma_m \vdash A_m$ where all subformulas
are of mode $m$ is directly a proof in the logic captured by the mode $m$.
Moreover, an arbitrary proof can be transformed into one of this form by cut
elimination.  These strong \emph{conservative extension} properties are a
hallmark of adjoint logic.

Since our main interest lies in natural deduction, we consider only three
examples.
\begin{example}[G3]
  We obtain the standard sequent calculus G3 \citep{Kleene52} for intutionistic
  logic with a single mode $\mU$.  All side conditions are automatically
  satisfied since $\mU \geq \mU$.
\end{example}

\begin{example}[LNL and DILL]
  By specializing the rules to two modes, $\mU$ and $\mL$ with the order
  $\mU > \mL$, we obtain a minor variant of LNL in its \emph{parsimonious
    presentation} \citep{Benton94tr}.  Our notation is $F X = \down^\mU_\mL X$
  and $G A = \up_\mL^\mU A$. Significant here is that we do not just model
  provability, but the exact structure of proofs except that our structural
  rules remain implicit.

  We obtain the sequent calculus formulation of dual intuitioninstic linear
  logic (DILL) \citep{Barber96,Chang03tr} by restricting the formulas of mode
  $\mU$ so that they only contain $\up_\mL^\mU A_\mL$.  In this version we have
  ${!}A = \down^\mU_\mL\, \up_\mL^\mU A$.  Again, the rules of dual
  intuitionistic linear logic are modeled precisely.
\end{example}

\begin{example}[Intuitionistic Subexponential Linear Logic]
  Subexponential linear logic \citep{Nigam09ppdp,Nigam16jlc} also uses a
  preorder of modes, each of which permits specific structural rules.  We obtain
  a formulation of \emph{intuitionistic subexponential linear logic} by adding a
  new distinguished mode $\mL$ with $m \geq \mL$ for all given subexponential
  modes $m$, retaining all the other relations.  We further restrict all modes
  $m$ except for $\mL$ to contain only $\up_\mL^m A_\mL$, forcing all logical
  inferences to take place at mode $\mL$.

  Compared to \citet{Chaudhuri10csl} our system does not contain ${?}A$ and is
  not focused; compared to \citet{Kanovich17arxiv}, our base logic is linear
  rather than ordered.  Also, all of our structural rules are implicit.
\end{example}

\section{Adjoint Natural Deduction}
\label{sec:nd}

Substructural \emph{sequent calculi} have recently found interesting
computational interpretations
\citep{Caires10concur,Wadler12icfp,Caires16mscs,Pfenning23coordination,Pruiksma22jfp},
including adjoint logic \citep{Pruiksma21jlamp}.  In this paper, we look instead
at \emph{functional} interpretations, which are most closely related to
\emph{natural deduction}.  Some guide is provided by natural deduction systems
for \emph{linear logic} (see, for example,
\citet{Abramsky93,Benton93tlca,Troelstra95apal}), but already they are not
entirely straightforward.  For example, some of these calculi do not satisfy
subject reduction. The interplay between modes and substructural properties
creates some further complications.  The closest blueprint to follow is probably
Benton's \citeyearpar[Figure 8]{Benton94tr}, but his system does not exhibit the
full generality of adjoint logic and is also not quite ``parsimonious'' in the
sense of the LNL sequent calculus.

In the interest of economy, we present the calculus with proof terms and two
bidirectional typing judgments, $\Delta \vdash e \chk A_m$ (expression $e$
checks against $A_m$) and $\Delta \vdash s \syn A_m$ (expression $s$
\emph{synthesizes} $A_m$).  The syntax for expressions can be found in
\autoref{fig:exps}, the rules in \autoref{fig:ndi}.  The bidirectional nature
will allow us to establish a precise relationship to the sequent calculus
(\autoref{sec:seqnd}), but it does not immediately yield a type checking
algorithm since the context merge operation is highly nondeterministic when used
to split contexts.  An algorithmic system can be found in \autoref{sec:algo}.

We obtain the vanilla typing judgment by replacing both checking and synthesis
judgments with $\Delta \vdash e : A$, dropping the rules
${\Rightarrow}/{\Leftarrow}$ and ${\Leftarrow}/{\Rightarrow}$, and removing the
syntactic form $(e : A_m)$.  We further obtain a pure natural deduction system
by removing the proof terms, although uses of the hypothesis rule then need to
be annotated with variables in order to avoid any ambiguities.

\begin{figure}[ht!]
\[
  \begin{array}{llcll}
    \mbox{Checkable Exps.} & e & ::= & \lambda x.\, e & (\lolli) \\
                           & & \mid & \{\ell \Rightarrow e_\ell\}_{\ell \in L} & (\with) \\
                           & & \mid & \mb{susp}\; e & (\up) \\[1ex]
                           & & \mid & (e_1, e_2) & (\tensor) \\
                           & & \mid & \mb{match}\; s\; ((x_1, x_2) \Rightarrow e') \\[1ex]
                           & & \mid & (\,) & (\one) \\
                           & & \mid & \mb{match}\; s\; ((\,) \Rightarrow e') \\[1ex]
                           & & \mid & \ell(e) & (\plus) \\
                           & & \mid & \mb{match}\; s\; (\ell(x) \Rightarrow e_\ell)_{\ell \in L} \\[1ex]
                           & & \mid & \mb{down}\; e & (\down) \\
                           & & \mid & \mb{match}\; s\; (\mb{down}\; x \Rightarrow e') \\[1ex]
                           & & \mid & s \\[1em]
    \mbox{Synthesizable Exps.} & s & ::= & x \\
                           & & \mid & s\; e & (\lolli) \\
                           & & \mid & s.\ell & (\with) \\
                           & & \mid & \mb{force}\; s & (\up) \\
                           & & \mid & (e : A_m) 
  \end{array}
\]
\caption{Expressions for Bidirectional Natural Deduction}
\label{fig:exps}
\end{figure}

\begin{figure}[ht!]
\begin{rules}
  \infer[{\Rightarrow}/{\Leftarrow}]
  {\Delta \vdash s \chk A_m}
  {\Delta \vdash s \syn A_m 
    }
  \hspace{3em}
  \infer[{\Leftarrow}/{\Rightarrow}]
  {\Delta \vdash (e : A_m) \syn A_m}
  {\Delta \vdash e \chk A_m}
  \\[1em]
  \infer[\ms{hyp}]{\Delta_{\mathsf{W}} \join x : A_m \vdash x \syn A_m}{}
  \\[1em]\hline\\[1ex]
  \infer[{\lolli}I]
  {\Delta \vdash \lambda x.\, e \chk A_m \lolli B_m}
  {\Delta, x:A_m \vdash e \chk B_m}
  \\[1em]
  \infer[{\lolli}E]
  {\Delta \join \Delta' \vdash s\, e \syn B_m}
  {\Delta \vdash s \syn A_m \lolli B_m
    & \Delta' \vdash e \chk A_m}
  \\[1em]
  \infer[{\with}I]
  {\Delta \vdash \{\ell \Rightarrow e_\ell\}_{\ell \in L} \chk {\with}\{\ell : A^\ell_m\}_{\ell \in L}}
  {\Delta \vdash e_\ell \chk A^\ell_m\quad (\forall \ell \in L)}
  \hspace{3em}
  \infer[{\with}E]
  {\Delta \vdash s.\ell \syn A_m^\ell}
  {\Delta \vdash s \syn {\with}\{\ell : A^\ell_m\}_{\ell \in L}
    \quad (\ell \in L)}
  \\[1em]
  \infer[{\up}I]
  {\Delta \vdash \mb{susp}\; e \chk \up_k^m A_k}
  {\Delta \vdash e \chk A_k}
  \hspace{3em}
  \infer[{\up}E]
  {\Delta_{\mathsf{W}} \join \Delta' \vdash \mb{force}\; s \syn A_k}
  {\Delta' \geq m & \Delta' \vdash s \syn \up_k^m A_k}
  \\[1em]\hline\\[1ex]
  \infer[{\tensor}I]
  {\Delta \join \Delta' \vdash (e_1,e_2) \chk A_m \tensor B_m}
  {\Delta \vdash e_1 \chk A_m
    & \Delta' \vdash e_2 \chk B_m}
  \\[1em]
  \infer[{\tensor}E]
  {\Delta \join \Delta' \vdash \mb{match}\; s\; ((x_1,x_2) \Rightarrow e') \chk C_r}
  {\Delta \vdash s \syn A_m \tensor B_m
    & \Delta \geq m \geq r
    & \Delta', x_1 : A_m, x_2 : B_m \vdash e' \chk C_r}
  \\[1em]
  \infer[{\one}I]
  {\Delta_{\mathsf{W}} \vdash (\,) \chk \one_m}
  {}
  \hspace{3em}
  \infer[{\one}E]
  {\Delta \join \Delta' \vdash \mb{match}\; s\; ((\,) \Rightarrow e') \chk C_r}
  {\Delta \vdash s \syn \one_m
    & \Delta \geq m \geq r
    & \Delta' \vdash e' \chk C_r}
  \\[1em]
  \infer[{\plus}I]
  {\Delta \vdash \ell(e) \chk {\plus}\{\ell : A_m^\ell\}_{\ell \in L}}
  {\Delta \vdash e \chk A_m^\ell}
  \\[1em]
  \infer[{\plus}E]
  {\Delta \join \Delta' \vdash \mb{match}\; s\; (\ell(x) \Rightarrow e_\ell)_{\ell \in L}
    \chk C_r}
  {\Delta \vdash s \syn {\plus}\{\ell : A_m^\ell\}_{\ell \in L}
    & \Delta \geq m \geq r
    & \Delta', x : A_m^\ell \vdash e_\ell \chk C_r\quad (\forall \ell \in L)}
  \\[1em]
  \hspace*{-1ex}
  \infer[{\down}I]
  {\Delta_{\mathsf{W}} \join \Delta' \vdash \mb{down}\; e \chk \down^n_m A_n}
  {\Delta' \geq n & \Delta' \vdash e \chk A_n}
  \hspace{1em}
  \infer[{\down}E]
  {\Delta \join \Delta'\vdash \mb{match}\; s\; (\mb{down}\; x \Rightarrow e') \chk C_r}
  {\Delta \vdash s \syn \down^n_m A_n
    & \Delta \geq m \geq r
    & \Delta', x : A_n \vdash e' \chk C_r}
\end{rules}
  \caption{Implicit Bidirectional Natural Deduction}
  \label{fig:ndi}
\end{figure}

The rules maintain a few important invariants, particularly
\emph{independence}:
\begin{enumerate}[(i)]
\item $\Delta \vdash e \chk A_m$ presupposes $\Delta \geq m$
\item $\Delta \vdash s \syn A_m$ presupposes $\Delta \geq m$
\end{enumerate}
This is somewhat surprising because we think of the synthesis judgment
$s \syn A_m$ as proceeding top-down rather than bottom-up.  Indeed, there are
other choices with dependence and structural properties being checked in
different places.  We picked this particular form because we want general typing
$e : A_m$ to arise from collapsing the checking/synthesis distinction.  This
means that the two rules ${\Rightarrow}/{\Leftarrow}$ and
${\Leftarrow}/{\Rightarrow}$ should have no conditions because those would
disappear.  The algorithmic system in \autoref{sec:algo} checks the
conditions in different places.

As an example of interesting rules we revisit $\down^n_m A_n$ (where $n \geq m$
is presupposed).  The introduction rule of natural deduction mirrors the
right rule of the sequent calculus, which is the case throughout.
\begin{rules}
  \infer[{\down}R]
  {\Gamma_{\mathsf{W}} \join \Gamma' \vdash \down^n_m A_n}
  {\Gamma' \geq n & \Gamma' \vdash A_n}
  \hspace{3em}
  \infer[{\down}I]
  {\Delta_{\mathsf{W}} \join \Delta' \vdash \mb{down}\; e \chk \down^n_m A_n}
  {\Delta' \geq n & \Delta' \vdash e \chk A_n}
\end{rules}
As is typical for these translations, the elimination rules turns the left rule
``upside down'' because (like all rules in natural deduction) the principal
formula is on the right-hand side of judgment, not the left as in the sequent
calculus.  This means we now have some conditions to check.
\begin{rules}
  \infer[{\down}L]
  {\Gamma \join \down^n_m A_n \vdash C_r}
  {\Gamma, A_n \vdash C_r}
  \hspace{3em}
  \infer[{\down}E]
  {\Delta \join \Delta'\vdash \mb{match}\; s\; (\mb{down}\; x \Rightarrow e') \chk C_r}
  {\Delta \vdash s \syn \down^n_m A_n
    & \Delta \geq m \geq r
    & \Delta', x : A_n \vdash e' \chk C_r}
\end{rules}
$\Delta \geq m$ is needed to enforce independence on the first premise.
$m \geq r$ together with $n \geq m$ enforces independence on the second premise.
Similar restrictions appear in the other elimination rules for the positive
connectives ($\tensor$, $\one$, $\plus$).

In general we see the following patterns in the correctness proofs below:
\begin{itemize}
\item The identity corresponds to ${\Rightarrow}/{\Leftarrow}$
\item Cut corresponds to ${\Leftarrow}/{\Rightarrow}$
\item Right rules correspond to introduction rules
\item Left rules correspond to upside-down elimination rules
  \begin{itemize}
  \item For negative connectives ($\lolli$, $\with$, $\up$) they
    are just reversed
  \item For positive connectives ($\tensor$, $\one$, $\plus$, $\down$)
    in addition a new hypothesis is introduced in a second premise
  \end{itemize}
\end{itemize}
From the last point we see that the hypothesis $x : A_m$ should just be read as
$x \syn A_m$.

We often say a natural deduction is \emph{normal}, which means that it cannot be
reduced, but under which collection of reductions?  The difficulty here is that
rewrite rules that reduce an introduction of a connective immediately followed
by its elimination are not sufficient to achieve deductions that are
\emph{analytic} in the sense that they satisfy the subformula property.  To
obtain analytic deductions, we have to add \emph{permuting conversions}.

We follow a different approach by directly characterizing \emph{verifications}
\citep{Dummett91,MartinLof83}, which are proofs that can be seen as constructed
by applying introduction rules bottom-up and elimination rules top-down.  By
definition, verifications satisfy the subformula property and are therefore
analytic and a suitable ``normal form'' even without defining a set of
reductions.

How does this play out here?  It turns out that if $\Delta \vdash e \chk A_m$
then the corresponding proof of $A_m$ (obtained by erasure of expressions) is a
\emph{verification} if the ${\Leftarrow}/{\Rightarrow}$ rule is disallowed and the
${\Rightarrow}/{\Leftarrow}$ rule is restricted to atomic propositions $P$.  By
our above correspondence that corresponds precisely to a cut-free sequent
calculus proof where the identity is restricted to atomic propositions.
Proof-theoretically, the meaning of a proposition is determined by its
\emph{verifications}, which, by definition, only decompose the given proposition
into its components.  Compare this with general proofs that do not obey such a
restriction.

In the next section we will prove that every proposition that has a proof also
has a verification by relating the sequent calculus and natural deduction.
\begin{example}[Church's ${\lambda}I$ calculus]
  \citet[Chapter {II}]{Church41book} introduced the ${\lambda}I$ calculus in
  which each bound variable requires at least one occurrence.  We obtain the
  simply-typed ${\lambda}I$ calculus with one mode $\mS$ with
  $\sigma(\mS) = \{\mathsf{C}\}$ and using $A_\mS \lolli B_\mS$ as the only type
  constructor.

  Similarly, we obtain the simply-typed $\lambda$-calculus with a single mode
  $\mU$ with $\sigma(\mU) = \{\mathsf{W},\mathsf{C}\}$ and the simply-typed \emph{linear}
  $\lambda$-calculus with a single mode $\mL$ with $\sigma(\mL) = \{\,\}$, using
  $A_\mL \lolli B_\mL$ as the only type constructor.
\end{example}

\begin{example}[Intuitionistic Natural Deduction]
  We obtain (structural) intutionistic natural deduction with a single mode
  $\mU$ with $\sigma(\mU) = \{\mathsf{W},\mathsf{C}\}$, where we can define
  $A \lor B = {\oplus}\{\mb{inl} : A, \mb{inr} : B\}$ and
  $\bot = {\oplus}\{\, \}$, $A \land B = {\with}\{\pi_1 : A, \pi_2 : B\}$ and
  $\top = {\with}\{\,\}$ and $A \to B = A \lolli B$.
\end{example}

\begin{example}[Intuitionistic S4]
  \label{ex:s4}
  We obtain the fragment of intuitionistic S4 in its dual formulation
  \citep{Pfenning01mscs} without possibility ($\Diamond A$) with two modes $\mV$
  and $\mU$ with $\mV > \mU$ and $\sigma(\mV) = \sigma(\mU) = \{\mathsf{W},\mathsf{C}\}$.  As in
  the DILL example of the adjoint sequent calculus, the mode $\mV$ is inhabited only by types $\up_\mU^\mV A_\mU$
  and we define $\Box A_\mU = \down^\mV_\mU \up_\mU^\mV A_\mU$, which is a
  \emph{comonad}.  The judgment $\Delta \semi \Gamma \vdash C\; \jtrue$ with
  valid hypotheses $\Delta$ and true hypothesis $\Gamma$ is modeled by
  $\Delta_\mV, \Gamma_\mU \vdash C_\mU$.

  The structure of \emph{verifications} is modeled almost exactly with one small
  exception: we allow the form $\Delta_\mV \vdash C_\mV$.  Because any
  proposition $B_\mV = \up_\mU^\mV A_\mU$, there is only one applicable rule to
  construct a verification of this judgment: ${\up}I$ (which, not
  coincidentally, is invertible).
\end{example}

\begin{example}[Lax Logic]
  We obtain natural deduction for lax logic \citep{Benton98,Pfenning01mscs} with
  two modes, $\mU$ and $\mX$, with $\mU > \mX$ and
  $\sigma(\mU) = \sigma(\mX) = \{\mathsf{W},\mathsf{C}\}$.  The mode $\mX$ is inhabited only by
  $\down^\mU_\mX A_\mU$.  We define
  $\bigcirc A_\mU = \up^\mU_\mX \down^\mU_\mX A_\mU$, which is a strong monad
  \citep{Benton98}.

  We model the rules of \citet{Pfenning01mscs} exactly, except that we allow
  hypotheses $B_\mX$, which must have the form $\down^\mU_\mX A_\mU$.  We can
  eagerly apply ${\down}E$ to obtain $A_\mU$, which again does not lose
  completeness by the invertibility of ${\down}L$ in the sequent calculus.

  We can also obtain linear versions of these relationships following
  \citep{Benton96}, although the term calculi do not match up exactly.
\end{example}

\section{Relating Sequent Calculus and Natural Deduction}
\label{sec:seqnd}

Rather than trying to find a complete set of proof reductions for natural
deduction, we translate a proof to the sequent calculus, apply cut and identity
elimination, and then translate the resulting proof back to natural deduction.
This is not essential, but it simultaneously proves the soundness and
completeness of natural deduction for adjoint logic and the completeness of
verifications.  This allows us to focus on the computational interpretation in
\autoref{sec:dynamics} that is a form of substructural functional programming.

For completeness of natural deduction, one might expect to prove that
$\Gamma \vdash C$ in the sequent calculus implies $\Gamma \vdash e \chk C$ in
natural deduction.  While this holds, a direct proof would not generate a
verification from a cut-free proof.  Intuitively, the way the proof proceeds
instead is to take a sequent $x_1 : A_1, \ldots, x_n : A_n \vdash C$ (ignoring
modes for the moment) and annotate each antecedent with a synthesizing term and
the succedent with an expression
$s_1 \syn A_1, \ldots, s_n \syn A_n \vdash e \chk C$.  This means we have to
account for the variables in $s_i$, and we do this with a substitution $\theta$
assigning synthesizing terms to each antecedent in $\Gamma$.  We therefore
define substitutions as mapping from variables to synthesizing terms.
\[
  \begin{array}{llcll}
    \mbox{Substitutions} & \theta & ::= & \cdot \mid \theta, x \mapsto s
  \end{array}
\]
We type substitutions with the judgment $\Delta \vdash \theta \syn \Gamma$,
where $\Delta$ contains the free variables in $\theta$.  This judgment
must respect independence and the structural properties of each
antecedent in $\Gamma$, as defined by the following rules:
\begin{rules}
  \infer[]
  {\cdot \vdash (\cdot) \ssyn (\cdot)}
  {}
  \hspace{3em}
  \infer[]
  {\Delta \join \Delta' \vdash (\theta, x \mapsto s) \syn (\Gamma, x : A_m)}
  {\Delta \vdash \theta \ssyn \Gamma
    & \Delta' \geq m
    & \Delta' \vdash s \syn A_m}
\end{rules}
We will use silently that if $\Delta \vdash \theta \ssyn \Gamma$ and $\Gamma \geq m$
then $\Delta \geq m$.

We write $e(x)$ and $s'(x)$ for terms with (possibly multiple, possibly no)
occurrences of $x$ and $e(s)$ and $s'(s)$ for the result of substituting $s$ for
$x$, respectively.  Because variables $x : A$ synthesize their types $x \syn A$,
the following admissible rules are straightforward assuming the premises satisfy
our presuppositions.

\begin{theorem}[Substitution Property]
  \label{thm:subst}
  \mbox{}
  The following are admissible:
  \begin{rules}
    \infer-[\ms{subst}/{\Leftarrow}]
    {\Delta \join \Delta' \vdash e(s) \chk C_{r}}
    {\Delta \vdash s \syn A_m
      & \Delta', x : A_m \vdash e(x) \chk C_{r}}
    \\[1em]
    \infer-[\ms{subst}/{\Rightarrow}]
    {\Delta \join \Delta' \vdash s'(s) \syn B_{k}}
    {\Delta \vdash s \syn A_m
      & \Delta', x : A_m \vdash s'(x) \syn B_{k}}
  \end{rules}
\end{theorem}
\begin{proof}
  By a straightforward simultaneous rule induction on the second given
  derivation.  In some cases we need to apply monotonicity.  For example, if $m$
  admits contraction and $\Delta \geq m$, then each hypothesis in $\Delta$ must
  also admit contraction.
\end{proof}

\begin{lemma}[Substitution Split]
  If $\Delta \vdash \theta \syn (\Gamma; \Gamma')$ then there exists $\theta_1$
  and $\theta_2$ and $\Delta_1$ and $\Delta_2$ such that
  $\Delta = \Delta_1 ; \Delta_2$ and $\Delta_1 \vdash \theta_1 \syn \Gamma$ and
  $\Delta_2 \vdash \theta_2 \syn \Gamma'$.
\end{lemma}
\begin{proof}
  By case analysis on the definition of context merge operation and induction on
  $\Delta \vdash \theta \syn (\Gamma; \Gamma')$. We rely on associativity and
  commutativity of context merge. We show two cases.

  \begin{description}
  \item[Case:] $(\Gamma_1, x:A_m) \join (\Gamma_2, x:A_m)
    = (\Gamma_1 \join \Gamma_2), x:A_m$  and $\mathsf{C} \in \sigma(m)$ 
    \[
      \infer
      {\Delta \join \Delta' \vdash (\theta_{12}, x \mapsto s) \syn (\Gamma_1\join\Gamma_2), x : A_m}
      {\Delta \vdash \theta_{12} \ssyn \Gamma_1\join\Gamma_2
        & \Delta' \geq k
        & \Delta' \vdash s \syn A_m}
    \]
    \begin{tabbing}
      $\Delta_1 \vdash \theta_1 \syn \Gamma_1$ and \\
      $\Delta_2 \vdash \theta_2 \syn \Gamma_2$ and \\
      $\Delta = \Delta_1 \join \Delta_2$ \` by IH
      \\
      $\Delta_1 \join \Delta' \vdash \theta_1, x \mapsto s \syn \Gamma_1, x : A_m$ \`
      by rule \\
      $\Delta_2 \join \Delta' \vdash \theta_2, x \mapsto s \syn \Gamma_2, x : A_m$ \`
      by rule \\
      since $\mathsf{C}\in \sigma(m)$ 
      and $\Delta' \geq m$, we have $\mathsf{C} \in \sigma(k)$ for any $B_k \in
      \Delta'$ \` by monotonicity
      \\
      $(\Delta_1 \join \Delta') \join (\Delta_2 \join \Delta') 
      = (\Delta_1 \join \Delta_2) \join \Delta' 
      = \Delta \join \Delta'$
      \` by previous line
    \end{tabbing}

  \item[Case:] $\Gamma_1 \join (\Gamma_2, x:A_m) = (\Gamma_1
    \join \Gamma_2), x:A_m$  and $x \not\in \rm{dom}(\Gamma_1)$ 
    \[
      \infer
      {\Delta \join \Delta' \vdash (\theta_{12}, x \mapsto s) \syn ((\Gamma_1\join\Gamma_2), x \mapsto A_m)}
      {\Delta \vdash \theta_{12} \ssyn \Gamma_1\join\Gamma_2
        & \Delta' \geq k
        & \Delta' \vdash s \syn A_m}
    \]

    \begin{tabbing}
      $\Delta_1 \vdash \theta_1 \syn \Gamma_1$ and \\
      $\Delta_2 \vdash \theta_2 \syn \Gamma_2$ and \\
      $\Delta = \Delta_1 \join \Delta_2$ \` by IH
      \\
      $\Delta_2 \join \Delta' \vdash \theta_2, x \mapsto s \syn \Gamma_2, x \mapsto A_m$ \`
      by rule 
      \\
      $\Delta_1 \join (\Delta_2 \join \Delta') 
      = (\Delta_1 \join \Delta_2) \join \Delta' 
      = \Delta \join \Delta'$
      \` by associativity of context merge
    \end{tabbing}
  \end{description}
\end{proof}

Now we have the pieces in place to prove the translation
from the sequent calculus to natural deduction.

\begin{theorem}[From Sequent Calculus to Natural Deduction]
  \label{thm:seq2nd}
  \mbox{}\newline
  If $\Gamma \vdash A_r$ and $\Delta \vdash \theta \ssyn \Gamma$ then
  $\Delta \vdash e \chk A_r$ for some $e$.
\end{theorem}
\begin{proof}
  By rule induction on the derivation $\DD$ of $\Gamma \vdash A_r$ and
  applications of inversion on the definition of substitution.  We present
  several indicative cases.  In this proof we write out the variables
  labeling the antecedents in sequents to avoid ambiguities.

  \begin{description}
  \item[Case:] $\DD$ ends in the identity.
    \[
      \DD =\quad  \infer[\ms{id}]{\Gamma_{\mathsf{W}} \join x : A_m \vdash A_m}{}
    \]
    \begin{tabbing}
      $\Delta \vdash \theta \ssyn (\Gamma_{\mathsf{W}} \join x : A_m)$ \` Given \\
      $\theta = (\theta_{\mathsf{W}}, x \mapsto s)$ \` By inversion \\
      $\Delta = (\Delta_{\mathsf{W}} \semi \Delta')$ with
      $\Delta_W \vdash \theta_{\mathsf{W}} \ssyn \Gamma_{\mathsf{W}}$ and
      $\Delta' \vdash s \syn A$ \` By context split \\
      $\Delta_{\mathsf{W}}$ satisfies weakening \` By monotonicity \\
      $\Delta' \vdash s \chk A$ \` By rule ${\Rightarrow}/{\Leftarrow}$ \\
      $\Delta_{\mathsf{W}} \semi \Delta' \vdash s \chk A$ \` By weakening \\
      $\Delta \vdash s \chk A$ \` Since $\Delta = (\Delta_{\mathsf{W}} \semi \Delta')$
    \end{tabbing}

  \item[Case:] $\DD$ ends in cut.
    \[
      \DD = \quad \infer[\ms{cut}]
      {\Gamma_1 \join \Gamma_2 \vdash C_r}
      {\Gamma_1 \geq m \geq r
        & \deduce[\DD_1]{\Gamma_1 \vdash A_m}{}
        & \deduce[\DD_2]{\Gamma_2, x : A_m \vdash C_r}{}}
    \]
    \begin{tabbing}
      $\Delta \vdash \theta \ssyn (\Gamma_1 \semi \Gamma_2)$ \` Given \\
      $\Delta = (\Delta_1 \join \Delta_2)$, $\theta = (\theta_1, \theta_2)$
      with $\Delta_1 \vdash \theta_1 \ssyn \Gamma_1$
      and $\Delta_2 \vdash \theta_2 \ssyn \Gamma_2$ \` By context split \\
      $\Delta_1 \vdash e_1 \chk A_m$ \` By IH on $\DD_1$ \\
      $\Delta_1 \vdash (e_1 : A_m) \syn A_m$ \` by rule ${\Leftarrow}/{\Rightarrow}$ \\
      $\Delta_2, x : A_m \vdash (\theta_2, x \mapsto x) \ssyn (\Gamma_2, x : A_m)$
      \` By subst. rule \\
      $\Delta_2, x : A_m \vdash e_2(x) \chk C_r$ \` By IH on $\DD_2$ \\
      $\Delta_1 \semi \Delta_2 \vdash e_2(e_1:A_m) \chk C_r$ \` By substitution (\autoref{thm:subst})
    \end{tabbing}

  \item[Case:] $\DD$ ends in $\up L$.
    \[ \DD = \quad
      \infer[{\up}L]
      {\Gamma \join x : \up_k^m A_k \vdash C_r}
      {\deduce[]{k \geq r}{} & \deduce[\DD']{\Gamma, y : A_k \vdash C_r}{}}
    \]
    We consider two subcases: $x : \up_k^m A_k \in \Gamma$ and $x : \up_k^m
    A_k \not\in \Gamma$. 
    \begin{description}
    \item[Subcase:] $x : \up_k^m A_k \in \Gamma$ and so
      $\mathsf{C} \in \sigma(m)$. Hence, $\Gamma \join x : \up_k^m A_k = \Gamma = \Gamma', x : \up_k^m A_k$.
      \begin{tabbing}
        $\Gamma \join x : \up^m_k A_k = \Gamma \geq r $ \` by presupposition \\
        $m \geq r$ \` since $(x : \up_k^m A_k)_m \in \Gamma$ \\
        $\Delta \vdash \theta \syn \Gamma $ \` by assumption\\ 
        $\Delta_1 \vdash \theta' \syn \Gamma'$ and 
        $\Delta_2 \vdash s \syn \up_k^m A_k$ and $\Delta_2 \geq m$ \` by
        inversion on subst. rule \\ 
        \` where $\Delta = \Delta_1 \join \Delta_2$ and $\theta
        = (\theta', x \mapsto s)$\\
        $\Delta_2 \vdash \mb{force}\; s \syn A_k$ \` by rule ${\up}E$ \\
        $\Delta \join \Delta_2 \vdash (\theta, y \mapsto \mb{force}\; s) \syn (\Gamma, y : A_k)$
        \` by subst. rule \\
        $\Delta \join \Delta_2 \vdash e \chk C_r$ for some $e$ \` by IH on $\DD'$ \\ 
        $\Delta \join \Delta_2 = \Delta$ \`  since $\mathsf{C}\in \sigma(m)$ 
        and $\Delta_2 \geq m$, we have $\mathsf{C} \in \sigma(j)$ for any $B_j \in
        \Delta_2$ \\
        $\Delta \vdash e \chk C_r$ \` by previous line 
      \end{tabbing}
    \item[Subcase:] $x : \up_k^m A_k \notin \Gamma$. 
      \begin{tabbing}
        $\Gamma, x : \up_k^m A_k \geq r $ \` by presupposition \\
        $\Gamma \geq r$ and $m \geq r$ \` by previous line \\
        $\Delta \vdash \theta \syn (\Gamma, x : \up_k^m A_k)$ \` by assumption\\ 
        $\Delta_1 \vdash \theta' \syn \Gamma$ and 
        $\Delta_2 \vdash s \syn \up_k^m A_k$ and $\Delta_2 \geq m$ \` by
        inversion on subst. rule \\ 
        \` where $\Delta = \Delta_1\join\Delta_2$ and $\theta
        = (\theta', x \mapsto s)$\\
        $\Delta_2 \vdash \mb{force}\; s \syn A_k$ \` by rule ${\up}E$ \\
        $\Delta_1 \join \Delta_2 \vdash (\theta', x \mapsto \mb{force}\; s) \syn (\Gamma, y : A_k)$
        \` by subst. rule \\
        $\Delta_1 \join \Delta_2 \vdash e \chk C_r$ \` by IH on $\DD'$ \\ 
        $\Delta \vdash e \chk C_r$ \` since $\Delta = (\Delta_1 \join \Delta_2)$
      \end{tabbing}
    \end{description}

  \item[Case:] $\DD$ ends in $\down L$.  Similarly to the previous case, we need
    to distinguish whether $x : \down^k_m A_k$ appears in $\Gamma$.
    \[ \DD = \quad
      \infer[{\down}L]
      {\Gamma \join x : \down^k_m A_k \vdash C_r}
      {\deduce[\DD']{\Gamma, y : A_k \vdash C_r}{}}
    \]
    We only show the case where $x : \down^k_m A_k$ is in $\Gamma$; the other
    follows similar reasoning.
    \begin{description}
    \item[Subcase:] $x : \down^k_m A_k \in \Gamma$ and so
      $\mathsf{C} \in \sigma(m)$. Hence, $\Gamma \join x : \down^k_m A_k = \Gamma = \Gamma', x : \down^k_m A_k$.
      \begin{tabbing}
        $\Gamma \geq r $ \` by presupposition \\
        $k \geq m \geq r$ \` by presuppositions \\
        $\Delta \vdash \theta \syn \Gamma \join x : \down^k_m A_k$ \` by assumption \\
        $\Delta_1 \vdash \theta' \syn \Gamma'$ and 
        $\Delta_2 \vdash s \syn \down^k_m A_k$ and $\Delta_2 \geq m$ \` by
        inversion on subst. rule \\ 
        \mbox{$\qquad$}\` where $\Delta = \Delta_1\join\Delta_2$ and $\theta
        = (\theta', x \mapsto s)$\\
        $y{:}A_k \vdash x \syn A_k$ \` by hyp \\
        $\Delta, y{:}A_k \vdash (\theta, x \mapsto x) \syn (\Gamma, y : A_k)$ \` by
        subst. rule (using $\Delta\join y{:}A_k = \Delta, y{:}A_k$)\\
        $\Delta,y{:}A_k \vdash e' \chk C_r$ for some $e'$ \` by IH on $\DD'$ \\
        $\Delta_2 \geq m \geq r$ \` by previous lines \\
        $\Delta_2 \join \Delta \vdash \mb{match}\; s\; (\mb{down}\; y \Rightarrow e') \chk C_r$
        \` by rule ${\down}E$ \\
        $\Delta \join \Delta_2 = \Delta$ \`  since $\mathsf{C}\in \sigma(m)$ 
        and $\Delta_2 \geq m$, we have $\mathsf{C} \in \sigma(j)$ for any $B_j \in
        \Delta_2$ \\
        $\Delta \vdash \mb{match}\; s\; (\mb{down}\; y \Rightarrow e') \chk C_r$ \` by previous line
      \end{tabbing}
    \end{description}

  \end{description}
\end{proof}

While there are no substitutions involved, the other direction has to take care
to introduce a cut \emph{only} for uses of the ${\Leftarrow}/{\Rightarrow}$
rule, and identity \emph{only} for uses of the ${\Rightarrow}/{\Leftarrow}$
rule.  This requires a generalization of the induction hypothesis so that the
elimination rules can be turned ``upside down''.

\begin{theorem}[From Natural Deduction to Sequent Calculus]
  \label{thm:nd2seq}
  \mbox{}
  \begin{enumerate}[(i)]
  \item If $\Delta \vdash e \chk C_r$ then $\Delta \vdash C_r$
  \item If $\Delta \vdash s \syn A_m$
    and $\Delta', x : A_m \vdash C_r$ 
    and $\Delta \geq r$ 
    then $\Delta \join \Delta' \vdash C_r$
  \end{enumerate}
\end{theorem}
\begin{proof}
  By simultaneous rule induction on $\Delta \vdash e \chk C_r$ and
  $\Delta \vdash s \syn A_m$.   We provide four sample cases.
    \begin{description}
    \item[Case:] The derivation ends in ${\Rightarrow}/{\Leftarrow}$.
      \[ \DD = \quad
        \infer[{\Rightarrow}/{\Leftarrow}]
        {\Delta \vdash s \chk A_m}
        {\deduce[\DD']{\Delta \vdash s \syn A_m}{}}
      \]
      \begin{tabbing}
        $x : A_m \vdash A_m$ \` By identity rule \\
        $\Delta \vdash A_m$ \` By IH(ii) with $\Delta' = (\cdot)$
      \end{tabbing}
    \item[Case:] The derivation ends in ${\Leftarrow}/{\Rightarrow}$.
      \[ \DD = \quad
        \infer[{\Leftarrow}/{\Rightarrow}]
        {\Delta \vdash (e : A_m) \syn A_m}
        {\deduce[\DD']{\Delta \vdash e \chk A_m}{}}
      \]
      \begin{tabbing}
        $\Delta', x : A_m \vdash C_r$ and $\Delta \geq r$ \` Assumption \\
        $\Delta \vdash A_m$ \` By IH(i) on $\DD'$ \\
        $\Delta \join \Delta' \vdash C_m$ \` By rule of cut 
      \end{tabbing}
  \item[Case:] The derivation ends in $\up E$
    \[ \DD = \quad
      \infer[{\up}E]{\Delta_{\mathsf{W}}\join\Delta \vdash \mb{force}\ s \syn A_k}{
        \deduce{\Delta \geq m}{} & 
        \deduce{\Delta \vdash s \syn \up_k^m A_k}{\DD'}
      }
    \]
    \begin{tabbing}
      $\Delta', x : A_k \vdash C_r$ with $\Delta', x : A_k \geq r$ \` by assumption  \\
      $\Delta_{\mathsf{W}} \join \Delta \geq k$ \` by presupposition \\ 
      $k \geq r$ \` since $\Delta',A_k \geq r$ \\
      $\Delta' \join y : \up_k^m A_k \vdash C_r$ \` by $\up L$ \\
      $\Delta_{\mathsf{W}} \join \Delta \geq r$ \` using previous lines \\
      $\Delta_{\mathsf{W}} \join \Delta \join \Delta'\vdash C_r$ \` by IH(ii)
    \end{tabbing}

  \item[Case:] The derivation ends in $\down E$
    \[ \DD = \quad
      \infer[\down E]{\Delta_1;\Delta_2 \vdash \mb{match}\ s\ (\mb{down}\ x \Rightarrow e') \chk C_r}{
        \deduce[\DD_1]{\Delta_1\vdash s \syn \down^k_m A_k}{} & 
        \deduce[]{\Delta_1 \geq m \geq r}{} &
        \deduce[\DD_2]{\Delta_2,x:A_k \vdash e' \chk C_r}{}
      }
    \]
    \begin{tabbing}
      $\Delta_2, x : A_k \vdash C_r$ \` by IH(i) on $\DD_2$ \\
      $\Delta_2 \join x : \down^k_m A_k \vdash C_r$ \` by rule $\down L$ \\
      $\Delta_1 \join \Delta_2 \vdash C_r$ \` by IH(ii) on $\DD_1$
    \end{tabbing}
  \end{description}
\end{proof}

As mentioned above, verifications are the foundational equivalent of normal
forms in natural deduction.  Using the two translations above we can show that
every provable proposition has a verification.  While we have not written the
translations out as functions, they constitute the computational contents of our
constructive proof of \autoref{thm:seq2nd} and \autoref{thm:nd2seq}.

\begin{theorem}
  If $\Delta \vdash e \chk A_m$ then there exists a verification
  of $\Delta \vdash e \chk A_m$.
\end{theorem}
\begin{proof}
  Given an arbitrary deduction of $\Delta \vdash e \chk A_m$, we can use
  \autoref{thm:nd2seq} (ii) to translate it to a sequent derivation of
  $\Delta \vdash A_m$.

  By the admissibility of cut and identity (\autoref{thm:cut-id}), we can obtain
  a long cut-free proof of $\Delta \vdash A_m$.

  We observe that the translation of \autoref{thm:seq2nd} translates 
  only cut to ${\Leftarrow}/{\Rightarrow}$ and only identity 
  to ${\Rightarrow}/{\Leftarrow}$.  Using the translation back
  to natural deduction from a long cut-free proof therefore results
  in a verification.
\end{proof}

\section{Dynamics}
\label{sec:dynamics}

As mentioned in \autoref{sec:nd}, we obtain a simple typing judgment
$\Delta \vdash e : A$ by collapsing the distinction between $e \chk A$ and
$s \syn A$, using $e$ as a universal notation for all expressions.  Furthermore,
the annotation $(e : A_m)$ is removed and the rules ${\Rightarrow}/{\Leftarrow}$
and ${\Leftarrow}/{\Rightarrow}$ are also removed.  The resulting rules remain
\emph{syntax-directed} in the sense that for every form of expression there is a
unique typing rule.

We further annotate the mode-changing constructors with the mode of their
subject, which in each case is uniquely determined by the typing derivation.
Some of these annotations are necessary, because the computation rules depend on
them; other information is redundant but kept for clarity.
\begin{itemize}
\item $\mb{susp}^m_k\; e : {\up}^m_k A_k$ if $e : A_k$
\item $\mb{force}^m_k\; e : A_k$ if $e : \up_k^m A_k$
\item $\mb{down}^n_m\; e : {\down}^n_m A_n$ if $e : A_n$
\item $\mb{match}_m\; e\; M_r : C_r$ if $e : A_m$
\end{itemize}
We give a sequential call-by-value semantics similar to the K machine (e.g.,
\citep[Chapter 28]{Harper16book}), but maintaining a global environment similar
to the Milner Abstract Machine \citep{Accattoli14icfp}.  There are two forms of
state in the machine:
\begin{itemize}
\item $\eta \semi K \rgt_m e$ (evaluate $e$ of mode $m$ under continuation stack $K$ and environment $\eta$)
\item $\eta \semi K \lft_m v$ (pass value $v$ of mode $m$ to continuation stack $K$ in environment $\eta$)
\end{itemize}
In the first, $e$ is an expression to be evaluated and $K$ is a stack of
continuations that the value of $e$ is passed to for further computation.  The
second then passes this value $v$ to the continuation stack.

The global environment $\eta$ maps variables to values, but these values may
again reference other variables.  In this way it is like Launchbury's
\citeyearpar{Launchbury93popl} heap, a connection we exploit in
\autoref{sec:recursion} to model call-by-need.  Because we maintain a global
environment, we do not need to build closures, nor do we need to substitute
values for variables.  Instead, we only (implicitly) rename variables to make
them globally unique.  This form of specification allows us to isolate the
\emph{dynamic use of variables}, which means we can observe the computational
consequences of modes and their substructural nature.  We could also use the
translation to the sequent calculus and then observe the consequence with an
explicit heap \citep{Pruiksma22jfp,Pfenning23coordination}, but in this paper we
study natural deduction and functional computation more directly.

\begin{figure}[ht!]
\[
  \begin{array}{llcll}
    \mbox{Frames} & f & ::= & \uscore\; e_2 \mid v_1\; \uscore & (\lolli) \\
                               & & \mid & \uscore.\ell & (\with) \\
                               & & \mid & \mb{force}^m_k\; \uscore & (\up) \\[1ex]
                               & & \mid & (\uscore, e_2) \mid (v_1, \uscore) & (\tensor) \\
                               & & \mid & \ell(\uscore) & (\plus) \\
                               & & \mid & \mb{down}^n_m\; \uscore & (\down) \\
                               & & \mid & \mb{match}_m\; \uscore\; M_r & (\tensor, \one, \plus, \down)
    \\[1em]
    \mbox{Continuations} & K & ::= & \epsilon \mid K \cdot f
    \\[1em]
    \mbox{Environments} & \eta & ::= & \cdot \mid \eta, x \mapsto v \mid \eta, [x \mapsto v]
    \\[1em]
    \mbox{Matches} & M & ::= & (x_1,x_2) \Rightarrow e'(x_1,x_2) & (\tensor) \\
                  & & \mid & (\,) \Rightarrow e' & (\one) \\
                  & & \mid & (\ell(x) \Rightarrow e_\ell(x))_{\ell \in L} & (\plus) \\
                  & & \mid & \mb{down}\; x \Rightarrow e'(x) & (\down)
    \\[1em]
    \mbox{Values} & v & ::= & \lambda x.\, e(x) & (\lolli) \\
                               & & \mid & \{\ell \Rightarrow e_\ell\}_{\ell \in L} & (\with) \\
                               & & \mid & \mb{susp}^m_k\; e & (\up) \\[1ex]
                               & & \mid & (v_1, v_2) & (\tensor) \\
                               & & \mid & (\,) & (\one) \\
                               & & \mid & \ell(v) & (\plus) \\
                               & & \mid & \mb{down}^n_m\; v & (\down)
    \\[1em]
    \mbox{States} & S & ::= & \eta \semi K \rgt_m e \\
    & & \mid & \eta \semi K \lft_m v
  \end{array}
\]
\caption{Machine States}
\label{fig:states}
\end{figure}

The syntax for continuations, environments, values, and machine states is
summarized in \autoref{fig:states}.  Although not explicitly polarized (as in
\citet{Levy06hosc}), values of negative type ($\lolli$, $\with$, $\up$) are lazy
in the sense that they abstract over unevaluated expressions, while values of
positive types ($\tensor$, $\one$, $\plus$, $\down$) are constructed from other
values.  This will be significant in our analysis of the computational
properties of modes.  Continuation frames just reflect the left-to-right
call-by-value nature of evaluation.

Values are typed as expressions.  Frames are typed with
$\Gamma \vdash f : B_k < A_m$, which means $f$ takes a value of type $A_m$ and passes
a value of type $B_k$ further up the continuation stack.  We show sample rules for
a negative ($\lolli$) and a positive ($\plus$) type.
\begin{rules}
  \infer[]
  {\Delta \vdash \uscore\, e_2 : B_m < (A_m \lolli B_m)}
  {\Delta \vdash e_2 : A_m}
  \hspace{3em}
  \infer[]
  {\Delta \vdash v_1\, \uscore : B_m < A_m}
  {\Delta \vdash v_1 : A_m \lolli B_m} 
  \\[1em]
  \infer[]
  {\Delta_{\mathsf{W}} \vdash \ell(\uscore) : {\plus}\{\ell : A^\ell_m\}_{\ell \in L} < A^\ell_m}
  {}
  \hspace{1em}
  \infer[]
  {\Delta \vdash \mb{match}_m\; \uscore\; (\ell(x) \Rightarrow e(x)) : C_r < {\plus}\{\ell : A^\ell_m\}_{\ell \in L}}
  {\Delta, x : A^\ell_m \vdash e(x) : C_r \quad (\forall \ell \in L)}
  \\[1em]\hline\\
  \infer[]
  {\Delta_{\mathsf{W}} \vdash \epsilon : A_m < A_m}
  {}
  \hspace{3em}
  \infer[]
  {\Delta \join \Delta' \vdash K \cdot f : C_r < A_m}
  {\Delta \vdash K : C_r < B_k & \Delta' \vdash f : B_k < A_m}
\end{rules}

Regarding environments we face a fundamental choice.  One possibility is to extend the
term language of natural deduction with explicit constructs for weakening and
contraction.  Then, similar to \citet{Girard87tapsoft}, no garbage collection
would be required during evaluation since uniqueness of references to variables would be maintained. 

We pursue here an alternative that leads to slightly deeper properties.  We
leave the structural rules implicit as in the rules so far.  This means that
variables of linear mode (that is, a mode that allows neither weakening nor
contraction) have uniqueness of reference and their bindings can be deallocated
when dereferenced.  Variables of structural mode (that is, a mode that allows
both weakening and contraction) are simply persistent in the dynamics and
therefore could be subject to an explicit garbage collection algorithm.

A difficulty arises with variables that only admit contraction but not
weakening.  After they are dereferenced the first time, they may or may not be
dereferenced again.  That is, they could be implicitly weakened after the first
access.  In order to capture this we introduce a new form of typing $[x : A_m]$
and binding $[x \mapsto v]$ we call \emph{provisional}.  A provisional binding
does not need to be referenced even if $m$ does not admit weakening.  The
important new property is that an ``ordinary'' variable $y : A_k$ that does not
admit weakening can not appear in a binding $[x \mapsto v]$.  In addition, all
the usual independence requirements have to be observed.

The rules for typing expressions, continuations, etc. are extended in the obvious
way, allowing variables $[x : A_m]$ to be used or ignored (as a part of some
$\Delta_W$).  We extend the context merge operation as follows, keeping in mind
that $x : A_m$ may require an occurrence of $x$ (depending on $\sigma(m)$),
while $[x : A_m]$ does not.
\[
  \begin{array}{rclcll}
    (\Delta_1, [x:A_m] & \join & (\Delta_2, [x:A_m]) & = & (\Delta_1 \join \Delta_2), [x:A_m]
    & \mbox{provided $\mathsf{C} \in \sigma(m)$} \\
    (\Delta_1, x:A_m) & \join & (\Delta_2, [x:A_m]) & = & (\Delta_1 \join \Delta_2), x:A_m
    & \mbox{provided $\mathsf{C} \in \sigma(m)$} \\
    (\Delta_1, [x:A_m] & \join & (\Delta_2, x:A_m) & = & (\Delta_1 \join \Delta_2), x:A_m
    & \mbox{provided $\mathsf{C} \in \sigma(m)$} \\
    (\Delta_1, [x:A_m]) & \join & \Delta_2 & = & (\Delta_1 \join \Delta_2), [x : A_m]
    & \mbox{provided $x \not\in \rm{dom}(\Delta_2)$} \\
    \Delta_1 & \join & (\Delta_2, [x:A_m]) & = & (\Delta_1 \join \Delta_2), [x : A_m]
    & \mbox{provided $x \not\in \rm{dom}(\Delta_1)$}
  \end{array}
\]
We have the following typing rules for environments.  $\Delta_W$ now means that
every declaration in $\Delta$ can be weakened, either explicitly because its
mode allows weakening, or implicitly because it is provisional.
\begin{rules}
  \infer[]
  {(\cdot) : (\cdot)}
  {}
  \hspace{1em}
  \infer[]
  {(\eta, x \mapsto v) : (\Delta, x : A_m)}
  {\eta : (\Delta \join \Delta') & \Delta' \geq m & \Delta' \vdash v : A_m}
  \hspace{1em}
  \infer[]
  {(\eta, [x \mapsto v]) : (\Delta, [x : A_m])}
  {\eta : (\Delta \join \Delta'_{\mathsf{W}}) & \Delta'_{\mathsf{W}} \geq m & \Delta'_{\mathsf{W}} \vdash v : A_m}
\end{rules}%
As an example, consider
$\eta_0 = (x \mapsto (\,), y \mapsto \lambda f.\, f\, x)$ where the mode of
variables is immaterial, but let's fix them to be $\mL$ with
$\sigma(\mL) = \{\,\}$.
\[
  \infer[]
  {(x \mapsto (\,), y \mapsto \lambda f.\, f\, x) : (y : (\one_\mL \lolli A_\mL) \lolli A_\mL)}
  {\infer[]
    {(x \mapsto (\,)) : (x : \one_\mL)}
    {\infer[]{(\cdot) : (\cdot)}{}
      & \infer[]{\cdot \vdash (\,) : \one_\mL}{}}
    & \deduce[\ddd]{x : \one_\mL \vdash \lambda f.\, f\, x : (\one_\mL \lolli A_\mL) \lolli A_\mL}{}}
\]
We observe that the binding of $x \mapsto (\,)$ does not contribute a
declaration $x : \one$ to the result context due to the occurrence of $x$ in the
value of $y$.

Now consider a slightly modified version where the mode of both $x$ and $y$ is
$\mS$ with $\sigma(\mS) = \{\mathsf{C}\}$, and the binding of $y \mapsto \ldots$
becomes provisional.
This modified example is no longer well-typed.
\[
  \infer[{??}]
  {(x \mapsto (\,), [y \mapsto \lambda f.\, f\, x]) : (y : [(\one_\mS \lolli A_\mS) \lolli A_\mS)]}
  {\infer[]
    {(x \mapsto (\,)) : (x : \one_\mS)}
    {\infer[]{(\cdot) : (\cdot)}{}
      & \infer[]{\cdot \vdash (\,) : \one_\mS}{}}
    & \deduce[\ddd]{x : \one_\mS \vdash \lambda f.\, f\, x : (\one_\mS \lolli A_\mS) \lolli A_\mS}{}}
\]
The problem is at the rule application marked by ${??}$.  The variable $y$ does
not need to be used, despite its mode, because the binding is provisional.  This
means that $x$ might also not be used because its only occurrence is in the value
of $y$.  But that is not legal, since the mode of $x$ does not admit weakening
and the binding is not provisional.

We type abstract machine states with the type of their final answer, that is
$s : C_r$.
\begin{rules}
  \infer[]
  {(\eta \semi K \rgt_m e) : C_r}
  {\eta : (\Delta \join \Delta') & \Delta \vdash K : C_r < A_m
    & \Delta' \vdash e : A_m}
  \\[1em]
  \infer[]
  {(\eta \semi K \lft_m v) : C_r}
  {\eta : (\Delta \join \Delta') & \Delta \vdash K : C_r < A_m
    & \Delta' \vdash v : A_m}
\end{rules}

We now continue with the computational rules for our abstract machine.  The
rules are in \autoref{fig:computation}.  We factor out passing a value to a
match, $\eta \semi v \pss_m M = \eta' \semi e'$ that produces a (possibly
extended) environment $\eta'$ and expression $e'$.  In all cases below, we
presuppose the variable names are chosen so the extended environment has unique
bindings for each variable.  For an extension with mutual recursion, see
\autoref{sec:recursion}.

\begin{figure}[ht!]
\[
  \begin{array}{lcrclclcl}
    \eta & \semi & (\,) & \pss_m & ((\,) \Rightarrow e') & = & \eta & \semi & e' \\
    \eta & \semi & (v_1,v_2) & \pss_m & ((x_1,x_2) \Rightarrow e'(x_1,x_2)) & = & \eta, x_1 \mapsto v_1, x_2 \mapsto v_2 & \semi & e'(x_1,x_2) \\
    \eta & \semi & \ell(v) & \pss_m & (\ell(x) \Rightarrow e'_\ell(x))_{\ell \in L} & = & \eta, x \mapsto v & \semi & e'_\ell(x) \\
    \eta & \semi & \mb{down}^n_m\, v & \pss_m & (\mb{down}(x) \Rightarrow e'(x)) & = & \eta, x \mapsto v & \semi & e'(x)
  \end{array}
\]
\[
  \begin{array}{lcll}
    \eta, x \mapsto v, \eta' \semi K \rgt_m x & \longrightarrow & \eta, \eta' \semi K \lft_m v 
    & (\mathsf{C} \not\in \sigma(m)) \\
    \eta, x \mapsto v, \eta' \semi K \rgt_m x & \longrightarrow & \eta, [x \mapsto v], \eta' \semi K \lft_m v 
    & (\mathsf{C} \in \sigma(m))
    \\[1em]
    \eta \semi K \rgt_r \mb{match}_m\; e\; M_r & \longrightarrow & \eta \semi K \cdot (\mb{match}_m\; \uscore\; M_r) \rgt_m e & ({\tensor}, {\one}, {\plus}, {\down}) \\
    \eta \semi K \cdot (\mb{match}_m\; \uscore\; M_r) \lft_m v
    & \longrightarrow
    & \multicolumn{2}{l}{\eta' \semi K \rgt_r e' \quad
    \mbox{where $\eta \semi v \pss_m M_r = \eta' \semi e'$}}
    \\[1em]
    \eta \semi K \rgt_m \lambda x.\, e(x) & \longrightarrow & \eta \semi K \lft_m \lambda x.\, e(x)
    & (\lolli) \\
    \eta \semi K \rgt_m (e_1\, e_2) & \longrightarrow & \eta \semi K \cdot (\uscore\, e_2) \rgt_m e_1 \\
    \eta \semi K \cdot (\uscore\; e_2) \lft_m v_1 & \longrightarrow & \eta \semi K \cdot (v_1\; \uscore) \rgt_m e_2 \\
    \eta nn\semi K \cdot (\lambda x.\, e(x), \uscore) \lft_m v_2 & \longrightarrow & \eta, x \mapsto v \semi K \rgt_m e(x)
    & 
    \\[1em]
    \eta \semi K \rgt_m \{\ell \Rightarrow e_\ell\}_{\ell \in L} & \longrightarrow & \eta \semi K \lft_m \{\ell \Rightarrow e_\ell\}_{\ell \in L}
    & (\with) \\
    \eta \semi K \rgt_m e.\ell & \longrightarrow & \eta \semi K \cdot (\uscore.\ell) \rgt_m e \\
    \eta \semi K \cdot (\uscore.\ell) \lft_m \{\ell \Rightarrow e_\ell\}_{\ell \in L} &
    \longrightarrow & \eta \semi K \rgt_m e_\ell & (\ell \in L)
    \\[1em]
    \eta \semi K \rgt_m \mb{susp}^m_k\; e & \longrightarrow & \eta \semi K \lft_m \mb{susp}^m_k\; e
    & (\up) \\
    \eta \semi K \rgt_k \mb{force}^m_k\; e & \longrightarrow & \eta \semi K \cdot (\mb{force}^m_k\; \uscore) \rgt_m e \\
    \eta \semi K \cdot (\mb{force}^m_k\; \uscore) \lft_m \mb{susp}^m_k e & \longrightarrow & \eta \semi K \rgt_k e
    \\[1em]
    \eta \semi K \rgt_m (e_1,e_2) & \longrightarrow & \eta \semi K \cdot (\uscore, e_2) \rgt_m e_1
    & (\tensor) \\
    \eta \semi K \cdot (\uscore, e_2) \lft_m v_1 & \longrightarrow & \eta \semi K \cdot (v_1, \uscore) \rgt_m e_2 \\
    \eta \semi K \cdot (v_1, \uscore) \lft_m v_2 & \longrightarrow & \eta \semi K \lft_m (v_1, v_2)
    \\[1em]
    \eta \semi K \rgt_m (\,) & \longrightarrow & \eta \semi K \lft_m (\,)
    & (\one)
    \\[1em]
    \eta \semi K \rgt_m \ell(e) & \longrightarrow & \eta \semi K \cdot \ell(\uscore) \rgt_m e
    & (\plus) \\
    \eta \semi K \cdot \ell(\uscore) \lft_m v & \longrightarrow & \eta \semi K \lft_m \ell(v)
    \\[1em]
    \eta \semi K \rgt_m \mb{down}^n_m\; e & \longrightarrow & \eta \semi K \cdot \mb{down}^n_m\; \uscore \rgt_n e
    & (\down) \\
    \eta \semi K \cdot (\mb{down}^n_m\; \uscore) \lft_n v & \longrightarrow
    & \eta \semi K \lft_m \mb{down}^n_m\; v
  \end{array}
\]
\caption{Computation Rules}
\label{fig:computation}
\end{figure}

We obtain the following expected theorems of preservation and progress.

\begin{theorem}[Preservation]
  If $S : A$ and $S \longrightarrow S'$ then $S' : A$.
\end{theorem}
\begin{proof}
  By cases on $S \longrightarrow S'$, applying inversion to the typing of $S$ and
  assembling a typing derivation of $S'$ from the resulting information.

  The trickiest case involves dereferencing a variable $x \mapsto v$ admitting
  contraction.  It is sound because every variable $y$ occurring in $v$ must
  also admit contraction by monotonicity and, furthermore, such variables still
  have an occurrence in the value $v$ that is being returned.  Therefore in the
  typing of the environment we can now type $[x \mapsto v]$ with $[x : A_m]$.
\end{proof}

A machine state is \emph{final} if it has the form
$\eta \semi \epsilon \lft_m v$, that is, if a value is returned to the empty
continuation in some global environment $\eta$.  In order to prove progress, we
need to characterize values of a given type using a \emph{canonical forms}
property.  Note that we allow a context $\Delta$ to provide for the variables
that may be embedded in a value of negative type ($\lolli$, $\with$, $\up$), but
that a variable by itself does not count as a value.

\begin{theorem}[Canonical Forms]
  If $\Delta \vdash v : A_m$ then one of the following applies:
  \begin{enumerate}[(i)]
  \item if $A_m = B_m \lolli C_m$ then $v = \lambda x.\, e(x)$ for some $e$
  \item if $A_m = {\with}\{\ell : A_m^\ell\}_{\ell \in L}$
    then $v = \{\ell \Rightarrow e_\ell\}_{\ell \in L}$ for some set $e_\ell$
  \item if $A_m = \up_k^m B_k$ then $v = \mb{susp}_k^m\; e$
  \item if $A_m = B_m \tensor C_m$ then $v = (v_1, v_2)$ for values $v_1$ and $v_2$
  \item if $A_m = \one$ then $v = (\,)$
  \item if $A_m = {\plus}\{\ell : B_m^\ell\}_{\ell \in L}$
    then $v = \ell(v')$ for some $\ell \in L$ and value $v'$
  \item if $A_m = \down^n_m A_n$ then $v = \mb{down}^n_m v'$ for some value $v'$
  \end{enumerate}
\end{theorem}

\begin{theorem}[Progress]
  If $S : C_r$ then
  either $S$ is final or $S \mapsto S'$ for some $S'$
\end{theorem}
\begin{proof}
  By cases on the typing derivation for the configuration and inversion on the
  typing of the embedded frames, values, and expressions.  We apply the
  canonical forms theorem when we need the shape of a value.
\end{proof}

Purely positive types play an important role because we view values of these
types as \emph{directly observable}, while values of negative types can only
be observed indirectly through their elimination forms.
\[
  \begin{array}{llcl}
    \mbox{Purely positive types}
    & A^+, B^+ & ::= & A^+ \tensor B^+ \mid \one \mid {\plus}\{\ell : A^+_\ell\}_{\ell \in L} \mid \down A^+
  \end{array}
\]
Values of purely positive types are closed, even if values of negative types may
not be.
\begin{lemma}[Positive Values]
  \label{lm:closed}
  If $\Delta \vdash v : A_r^+$ then $\cdot \vdash v : A_r^+$ and all
  declarations in $\Delta$ admit weakening (either due to their mode or because
  they are provisional).
\end{lemma}
\begin{proof}
  By induction on the structure of the typing derivation, recalling that
  variables are not values.
\end{proof}

We call a variable $x : A_m$ \emph{linear} if $\sigma(m) = \{\,\}$, that is, the
mode $m$ admits neither weakening or contraction.  We extend this term to types,
bindings in the environment, etc. in the obvious way.

\begin{theorem}[Freedom from Garbage]
  If $\cdot \vdash e : A_r^+$ and
  $\cdot \semi \epsilon \rgt_r e \longrightarrow^* \eta \semi \epsilon \lft_r v$, then
  $\eta$ does not contain a binding $x \mapsto v$ with $\sigma(m) = \{\,\}$.
\end{theorem}
\begin{proof}
  Because $A_r^+$ is purely positive, we know by \autoref{lm:closed} that $v$ is
  closed.

  When the continuation $K$ is empty, the typing rule for valid states
  implies that $\eta : \Delta$ and $\Delta \vdash v : A_r^+$ for some
  $\Delta$.  Since $v$ is closed, $\Delta$ cannot contain any linear
  variables.

  Then we prove by induction on the typing of $\eta$ that none of
  variables in $\eta$ can be linear.  In the inductive case
  \begin{rules}
    \infer[]
    {(\eta', x \mapsto v) : (\Delta, x : A_m)}
    {\eta' : (\Delta \join \Delta') & \Delta' \geq m & \Delta' \vdash v : A_m}
  \end{rules}%
  we know that $m$ must admit weakening or contraction or both.  Since
  $\Delta' \geq m$, by monotonicity, $\Delta'$ must also admit weakening or
  contraction and we can apply the induction hypothesis to
  $\eta' : (\Delta \join \Delta')$.
\end{proof}

We call a variable $x_m$, an expression $e : A_m$, or a binding $x \mapsto v$
\emph{strict} if $\sigma(m) \subseteq \{\mathsf{C}\}$, that is, $m$ does not admit
weakening.

\begin{theorem}[Strictness]
  If $\cdot \vdash e : A_r^+$ and
  $\cdot \semi \epsilon \rgt_r e \longrightarrow^* \eta \semi \epsilon \lft_r v$, then
  every \emph{strict} binding in $\eta$ is of the form $[x \mapsto v]$.
\end{theorem}
\begin{proof}
  Because $A_r^+$ is purely positive, we know by \autoref{lm:closed}
  that $v$ is closed.

  When the continuation $K$ is empty, the typing rule for valid states implies
  $\eta : \Delta$ and $\Delta \vdash v : A_r^+$ for some $\Delta$.  Since $v$ is
  closed, $\Delta$ contains strict variables only in the form $[x : A_m]$.

  We prove by induction on the typing of $\eta$ all strict variables in $\eta$
  have the form $[x \mapsto w]$.  There are two inductive cases.
  \begin{rules}
    \infer[]
    {(\eta', x \mapsto w) : (\Delta, x : A_m)}
    {\eta' : (\Delta \join \Delta') & \Delta' \geq m & \Delta' \vdash w : A_m}
  \end{rules}%
  Since $m$ is not strict, it must admit weakening.  Since
  $\Delta' \geq m$, every variable in $\Delta'$ must also admit
  weakening by monotonicity, so we can apply the induction
  hypothesis to $\Delta \join \Delta'$.

  \begin{rules}
    \infer[]
    {(\eta', [x \mapsto v]) : (\Delta, [x : A_m])}
    {\eta' : (\Delta \join \Delta'_{\mathsf{W}}) & \Delta'_{\mathsf{W}} \geq m & \Delta'_{\mathsf{W}} \vdash w : A_m}
  \end{rules}%
  Any declaration in $\Delta'_{\mathsf{W}}$ either directly admits weakening or is of
  the form $[y : A_k]$ for a strict $k$ so we can apply the induction hypothesis
  to $\eta' : (\Delta \join \Delta'_{\mathsf{W}})$.
\end{proof}

In this context of call-by-value, this property expresses that every strict
variable will be read at least once, since a binding $[x \mapsto v]$ arises
only from reading the value of $x$.




\begin{theorem}[Dead Code]
  If $\cdot \vdash e : A_r^+$ and
  $\cdot \semi \epsilon \rgt_r e \longrightarrow^* \eta \semi \epsilon \lft_r v$
  then every state during the computation either evaluates $\rgt_m$ or
  returns $\lft_m$ for $m \geq r$.
\end{theorem}
\begin{proof}
  Most rule do not change the subject's mode.  Several rules potentially
  raise the mode, name evaluating a $\mb{match}$, a $\mb{force}$, or a $\mb{down}$.
  For each of these there is a corresponding rule lowering the mode back to
  its original, namely return a value to a $\mb{match}$, to a $\mb{force}$,
  or to a $\mb{down}$.

  We say the mode of a frame $f$ is the mode of the following state after a
  value is returned to $f$.  We prove by induction over the computation that in
  all states, all continuation frames and subjects have modes $m \geq r$.
\end{proof}

\begin{corollary}[Erasure]
  Assume $\cdot \vdash e : A_r^+$ and
  $\cdot \semi \epsilon \rgt_r e \longrightarrow^* \eta \semi \epsilon \lft_r
  v$.  Let $\Omega$ be a new term of every type and no transition rule.
  
  If we obtain $e'$ by replacing all subterms $e' : B_k$ for $k \not\geq r$ with
  $\Omega$, then evaluation $e'$ still terminates in a final state.  This final
  state differs from $v$ in that subterms of mode $k \not\geq r$ are also
  replaced by $\Omega$.
\end{corollary}
\begin{proof}
  The computation of $e'$ parallels that of $e$.  It would only get stuck
  for a state $\eta' \semi K' \rhd_k \Omega$, but that is impossible
  by the preceding dead code theorem since $k \not\geq r$.
\end{proof}

\section{Adding Recursion and Call-by-Need}
\label{sec:recursion}

To add recursion at the level of types we allow a global environment $\Sigma$ of
equirecursive type definition $t_m = A_m$ that may be mutually recursive and
restrict $A_m$ to be contractive, that is, start with a type constructor and not
a type name.  Since we view them as equirecursive we can silently unfold them,
and the only modification is in the rule for directional change.
The typing judgment is now parameterized by a signature $\Sigma$, but since it
is fixed we generally omit it.
\begin{rules}
  \infer[{\Rightarrow}/{\Leftarrow}]
  {\Delta \vdash s \chk A_m}
  {\Delta \vdash s \syn A_m' & A_m' = A_m}
\end{rules}%
Type equality follows the standard coinductive definition \citep{Amadio93toplas}.  It
can also be replaced by subtyping $A_m' \leq A_m$ without affecting soundness,
but a formal treatment is beyond the scope of this paper.

In order to model recursion at the level of expressions, we allow
\emph{contextual definitions} \citep{Nanevski08tocl} in the signature.  They
have the form $f [\Delta] : A_m = e$ and may be mutually recursive, so all
definitions are checked with respect to the same global signature $\Sigma$,
using the $\vdash_\Sigma \Sigma'\; \mathit{sig}$ judgment.  We use here a form
of checkable substitution $\Delta' \vdash \theta \chk \Delta$ which we did not
require earlier, and which has the obvious pointwise definition.
\begin{rules}
  \infer[]
  {\vdash_\Sigma \Sigma', t_m = A_m}
  {\vdash_\Sigma \Sigma'\; \mathit{sig}
    & \vdash_\Sigma A_m \; \mathit{type}
    & A_m\; \mathit{contractive}}
  \hspace{1em}
  \infer[]
  {\vdash_\Sigma \Sigma', f[\Delta] : A_m = e\; \mathit{sig}}
  {\vdash_\Sigma \Sigma' \; \mathit{sig}
    & \Delta \vdash_\Sigma e \chk A_m}
  \hspace{1em}
  \infer[]
  {\vdash_\Sigma (\cdot)\; \mathit{sig}}
  {}
  \\[1em]
  \infer[\ms{call}]
  {\Delta' \vdash_\Sigma f[\theta] \syn A_m}
  {f [\Delta] : A_m = e \in \Sigma
    & \Delta' \vdash_\Sigma \theta \chk \Delta}
\end{rules}%
The reason that definitions are contextual is partly pragmatic.  We are used to
reifying, say, $f\, (x : A)\, (y : B) : C = e(x,y)$ as
$f : A \arrow B \arrow C = \lambda x.\, \lambda y.\, e(x,y)$.  But because
$A_m \lolli B_m$ requires the same mode on both sides, we would have to insert
shifts, or generalize the function type and thereby introduce multiple
mode-shifting constructors.  By using contextual definitions, we can directly
express dependence on a mixed context.  For example,
$f\, [x : A_m, y : B_k] : C_r = e(x,y)$ which, by presupposition, requires
$m,k \geq r$ so that the judgment $x : A_m, y : B_k \vdash e(x,y) \chk C_r$ is
well-formed.

In order to demonstrate call-by-need, let's assume that contextual definitions
are indeed treated in call-by need manner.  In order to model this, we create a
new form of environment entry $x \mapsto e$ where $e$ is an unevaluated
expression.  A substitution
$\theta = (x_1 \mapsto e_1, \ldots, x_n \mapsto e_n)$ then also represents an
environment in which none of the expressions $e_i$ depend on any of the
variables $x_i$.  We also need to new continuation frame $x \mapsto \uscore$ to
create a value binding in the environment.  We have the following new rules:
\[
  \begin{array}{lcll}
    \eta \semi K \rgt_m f[\theta]
    & \longrightarrow
    & \eta, \theta \semi K \rgt_m e \qquad
      \mbox{for $f[\Delta]:A_m = e \in \Sigma$}
    \\[1ex]
    \eta, x \mapsto e, \eta' \semi K \rgt_m x
    & \longrightarrow
    & \eta, [x \mapsto e], \eta' \semi K \cdot (x \mapsto \uscore) \rgt_m e
    \\
    \eta, [x \mapsto e], \eta' \semi K \cdot (x \mapsto \uscore) \lft_m v
    & \longrightarrow
    & \eta, [x \mapsto v], \eta' \semi K \lft_m v
  \end{array}
\]
We can further simplify these rules if $m$ does \emph{not} admit contraction,
dropping $[x \mapsto e]$ and $[x \mapsto v]$.  We could also model ``black
holes'' by rebinding not $[x \mapsto e]$ in the second rule, but
$[x \mapsto \Omega]$ where $\Omega$ has every type, but does not reduce.

The theorems from the previous section continue to hold.  They imply that in a
final state of observable type there will be no binding $x \mapsto e$ in the
environment if the mode of $x$ is strict (that is, does not admit contraction).
In other words, strict expressions are always evaluated.

\section{Algorithmic Type Checking}
\label{sec:algo}

The bidirectional type system of \autoref{sec:nd} is not yet algorithmic, among
other things because splitting a given context into
$\Delta = (\Delta_1 \semi \Delta_2)$ is nondeterministic.  One standard solution
is track which are hypotheses are used in one premise (which ends up
$\Delta_1$), subtract them from the available ones, and pass the remainder into
the second premise (which ends up $\Delta_2$ together with an overall remainder)
\citep{Cervesato00tcs}.  This originated in proof search, but here when we
actually have a proof terms available to check, other options are available.
\emph{Additive} resource management computes the \emph{used} hypotheses (rather
than the unused ones) and merges (``adds'') them \citep{Atkey18lics}, which is
conceptually slightly simpler and also has been shown to be more efficient
\citep{Hughes20lopstr}.

The main complication in the additive approach are internal and external choice,
more specifically, the ${\with}R$ and ${\plus}L$ rules when the choice is empty.
For example, while check $\Delta \vdash \{\,\} \chk {\with}\{\,\}$ any subset of
$\Delta$ could be used.  We reuse the idea from the dynamics to have hypothesis
$[x : A_m]$ which may or may not have been used.  Unfortunately, this leads
to a plethora of different judgments, but it is not clear how to simplify them.
In defining the additive approach, the main judgment is
\[
  \Gamma \vdash e \Longleftrightarrow A_m \uses \Xi
\]
where $\Gamma$ is a plain (that is, free of provisional hypotheses) context
containing \emph{all} variables that might occur in $e$ (regardless of mode or
structural properties) and $\Xi$ is a context that may contain provisional
hypotheses.  We maintain the mode invariant $\Xi \geq m$ (even if it may be the
case that $\Gamma \not\geq m$).

Because we keep the contexts $\Delta$ free of provisional hypotheses,
we define the relation $\Xi \sqsupseteq \Delta$ which may remove
or keep provisional hypotheses.
\[
  \begin{array}{rclcrcl}
    (\Xi, x : A_m) & \sqsupseteq & (\Delta, x : A) & \mbox{if} & \Xi & \sqsupseteq & \Delta \\
    (\Xi, [x : A_m]) & \sqsupseteq & (\Delta, x : A) & \mbox{if} & \Xi & \sqsupseteq & \Delta \\
    (\Xi, [x : A_m]) & \sqsupseteq & \Delta & \mbox{if} & \Xi & \sqsupseteq & \Delta \\
    (\cdot) & \sqsupseteq & (\cdot)
  \end{array}
\]
With this relation, we can state the \emph{soundness} of algorithmic typing (postponing
the proof).
\begin{theorem}[Soundness of Algorithmic Typing]
  \label{thm:algo-sound}\mbox{}\newline
  If $\Gamma \vdash e \Longleftrightarrow A_m \uses \Xi$ and $\Xi \sqsupseteq \Delta$
  then $\Delta \vdash e \Longleftrightarrow A_m$.
\end{theorem}

For completeness we need a different relation $\Delta \geq \Xi$ which
means that $\Xi$ contains a legal subset of the hypotheses in $\Delta$.
This means hypotheses in $\Delta$ might be in $\Xi$ (possibly provisional)
or not, but then only if they can we weakened.
\[
  \begin{array}{rclcrcll}
    (\Delta, x : A_m) & \geq & (\Xi, x : A_m) & \mbox{if} & \Delta & \geq & \Xi \\
    (\Delta, x : A_m) & \geq & (\Xi, [x : A_m]) & \mbox{if} & \Delta & \geq & \Xi \\
    (\Delta, x : A_m) & \geq & \Xi & \mbox{if} & \Delta & \geq & \Xi & \mbox{provided $\mathsf{W} \in \sigma(m)$} \\
    (\cdot) & \geq & (\cdot)
  \end{array}
\]
With this relation we can state the \emph{completeness} of algorithmic typing
(also postponing the proof).

\begin{theorem}[Completeness of Algorithmic Typing]
  \label{thm:algo-complete}\mbox{}\newline
  If $\Delta \vdash e \Longleftrightarrow A_m$ then
  $\Delta \vdash e \Longleftrightarrow A_m \uses \Xi$ for some $\Xi$ with $\Delta \geq \Xi$
\end{theorem}

For the algorithm itself we need several operations.  The first,
$\Xi \ctxrm x : A$ removes $x : A$ from $\Xi$ if this is legal operation.
Its prototypical use is in the ${\lolli}I$ rule:
\begin{rules}
  \infer[{\lolli}I]
  {\Gamma \vdash \lambda x.\, e \chk A_m \lolli B_m \uses (\Xi \ctxrm x:A_m)}
  {\Gamma, x:A_m \vdash e \chk B_m \uses \Xi}
\end{rules}
For this rule application to be correct, $x : A_m$ must either have been used
and therefore occur in $\Xi$, or the mode $m$ must allow weakening.
\[
  \begin{array}{rclcll}
    (\Xi, x : A_m) & \ctxrm & x : A_m & = & \Xi \\
    (\Xi, [x : A_m]) & \ctxrm & x : A_m & = & \Xi \\
    (\Xi, y : B_k) & \ctxrm & x : A_m & = & (\Xi \ctxrm x : A_m), y : B_k & \mbox{provided $y \neq x$} \\
    (\Xi, [y : B_k]) & \ctxrm & x : A_m & = & (\Xi \ctxrm x : A_m), y : B_k & \mbox{provided $y \neq x$} \\
    (\cdot) & \ctxrm & x : A & = & (\cdot) & \mbox{provided $\mathsf{W} \in \sigma(m)$}
  \end{array}
\]
We have the following fundamental property (modulo exchange on contexts, as
usual).
\begin{lemma}
  If $\Xi \ctxrm x : A_m = \Xi'$ then
  \begin{enumerate}[(i)]
  \item either $\Xi = (\Xi', x : A_m)$
  \item or $\Xi = (\Xi', [x : A_m])$
  \item or $\Xi = \Xi'$ with $x \not\in \mathrm{dom}(\Xi)$ and $\mathsf{W} \in \sigma(m)$.
  \end{enumerate}
\end{lemma}
We also need two forms of context restriction.  The first $\Xi\bbar_m$ removes
all hypotheses whose mode is \emph{not} greater or equal to $m$ to restore our
invariant.  It fails if $\Xi$ contains a \emph{used} hypothesis $B_r$ with
$r \not\geq m$.  It is used only in the ${\up}I$ rule to restore the invariant,
\begin{rules}
  \infer[{\up}I]
  {\Gamma \vdash \mb{susp}\; e \chk \up_k^m A_k \uses \Xi\bbar_m}
  {\Gamma \vdash e \chk A_k \uses \Xi}
\end{rules}
\[
  \begin{array}{rcll}
    (\Xi, x:A_k)\bbar_m & = & \Xi\bbar_m, x:A_k & \mbox{provided $k \geq m$} \\
    (\Xi, [x:A_k])\bbar_m & = & \Xi\bbar_m, [x:A_k] & \mbox{provided $k \geq m$} \\
    (\Xi, [x:A_k])\bbar_m & = & \Xi\bbar_m & \mbox{provided $k \not\geq m$} \\
    (\cdot)\bbar_m & = & (\cdot)
  \end{array}
\]
It has the following defining property.
\begin{lemma}\mbox{}
  If $\Delta \geq m$ and $\Delta \geq \Xi$ then $\Delta \geq \Xi\bbar_m$.
\end{lemma}

The second form of context restriction occurs in the case of an
empty internal or external choice.  All of the hypothesis that are
allowed by the independence principle could be considered used, but
they might also not.  We write $[\Gamma|_m]$ and define:
\[
  \begin{array}{rcll}
    [(\Gamma, x : A_k)|_m] & = & [\Gamma|_m], [x : A_k] & \mbox{provided $k \geq m$} \\\relax
    [(\Gamma, x : A_k)|_m] & = & [\Gamma|_m] & \mbox{provided $k \not\geq m$} \\\relax
    [(\cdot)|_m] & = & (\cdot)
  \end{array}
\]
It is used only in the nullary case for internal and external choice.
\begin{rules}
  \infer[{\with}I_0]
  {\Gamma \vdash \{\,\} \chk {\with}_m\{\,\} \uses [\Gamma|_m]}
  {}
  \hspace{3em}
  \infer[{\plus}E_0]
  {\Gamma \vdash \mb{match}\; s\; (\,) \chk C_r \uses \Xi \semi [\Gamma|_r]}
  {\Gamma \vdash s \syn \oplus_m\{\,\} \uses \Xi
    & m \geq r}
\end{rules}
We come to the final operation $\Xi_1 \lub \Xi_2$ which is needed
for ${\with}I$ and ${\plus}E$.  We show the case of ${\with}I_2$,
that is, the binary version.
\begin{rules}
  \infer[{\with}I_2]
  {\Gamma \vdash \{\pi_1 \Rightarrow e_1, \pi_2 \Rightarrow e_2\}
    \chk {\with}\{\pi_1 : A_m, \pi_2 : B_m\} \uses \Xi_1 \lub \Xi_2}
  {\Gamma \vdash e_1 \chk A_m \uses \Xi_1
    & \Gamma \vdash e_2 \chk B_m \uses \Xi_2}
\end{rules}
Variables $y : B_k$ that are definitely used in $\Xi_1$ or $\Xi_2$ must also
be used in the other, or be available for weakening.  This could be
because they are provisional $[y : B_k]$ or because mode $k$ admits
weakening.  This idea is captured by the following definition.
\[
  \begin{array}{rclcll}
    (\Xi_1, x:A_m) & \lub & (\Xi_2, x:A_m) & = & (\Xi_1 \lub \Xi_2), x : A_m \\
    (\Xi_1, [x:A_m]) & \lub & (\Xi_2, x:A_m) & = & (\Xi_1 \lub \Xi_2), x : A_m \\
    (\Xi_1, x:A_m) & \lub & (\Xi_2, [x:A_m]) & = & (\Xi_1 \lub \Xi_2), x : A_m \\
    (\Xi_1, [x:A_m]) & \lub & (\Xi_2, [x:A_m]) & = & (\Xi_1 \lub \Xi_2), [x : A_m] \\
    (\Xi_1, x:A_m) & \lub & \Xi_2 & = & (\Xi_1 \lub \Xi_2), x:A_m 
                      & \mbox{for $x \not\in \rm{dom}(\Xi_2)$, $\mathsf{W} \in \sigma(m)$} \\
    \Xi_1 & \lub & (\Xi_2, x:A_m) & = & (\Xi_1 \lub \Xi_2), x:A_m 
                      & \mbox{for $x \not\in \rm{dom}(\Xi_1)$, $\mathsf{W} \in \sigma(m)$} \\
    (\Xi_1, [x : A_m]) & \lub & \Xi_2 & = & \Xi_1 \lub \Xi_2 
                        & \mbox{for $x \not\in \rm{dom}(\Xi_2)$} \\
    \Xi_1 & \lub & (\Xi_2, [x : A_m]) & = & \Xi_1 \lub \Xi_2 
                        & \mbox{for $x \not\in \rm{dom}(\Xi_1)$} \\
    (\cdot) & \lub & (\cdot) & = & (\cdot)
  \end{array}
\]
We have the following two properties, expressing the property
of a least upper bound in two slightly asymmetric judgments.
\begin{lemma}\mbox{}
  \begin{enumerate}[(i)]
  \item If $\Delta \geq \Xi_1$ and $\Delta \geq \Xi_2$ then $\Delta \geq \Xi_1 \lub \Xi_2$.
  \item if $\Xi_1 \lub \Xi_2 \sqsupseteq \Delta$ then for some $\Delta_1$, $\Delta^1_{\mathsf{W}}$,
    $\Delta_2$, $\Delta^2_{\mathsf{W}}$ we have $\Xi_1 \sqsupseteq \Delta_1$, $\Xi_2 \sqsupseteq \Delta_2$
    and $\Delta = \Delta_1 \join \Delta_{\mathsf{W}}^1 = \Delta_2 \join \Delta_{\mathsf{W}}^2$.
  \end{enumerate}
\end{lemma}

The complete set of rules for algorithmic typing can be found in \autoref{fig:algo}.

\begin{figure}[ht!]
\begin{rules}
  \infer[{\Rightarrow}/{\Leftarrow}]
  {\Gamma \vdash s \chk A_m \uses \Xi}
  {\Gamma \vdash s \syn A_m \uses \Xi}
  \hspace{1em}
  \infer[{\Leftarrow}/{\Rightarrow}]
  {\Gamma \vdash (e : A_m) \syn A_m \uses \Xi}
  {\Gamma \vdash e \chk A_m \uses \Xi}
  \hspace{1em}
  \infer[\ms{hyp}]
  {\Gamma \vdash x \syn (x : A_m)}
  {x : A_m \in \Gamma}
  \\[1em]
  \infer[{\lolli}I]
  {\Gamma \vdash \lambda x.\, e \chk A_m \lolli B_m \uses (\Xi \ctxrm x:A_m)}
  {\Gamma, x:A_m \vdash e \chk B_m \uses \Xi}
  \\[1em]
  \infer[{\lolli}E]
  {\Gamma \vdash s\, e \syn B_m \uses \Xi \join \Xi'}
  {\Gamma \vdash s \syn A_m \lolli B_m \uses \Xi
    & \Gamma \vdash e \chk A_m \uses \Xi'}
  \\[1em]
  \infer[{\with}I_0]
  {\Gamma \vdash \{\,\} \chk {\with}_m\{\,\} \uses [\Gamma|_m]}
  {}
  \hspace{2em}
  \infer[{\with}I]
  {\Gamma \vdash \{\ell \Rightarrow e_\ell\}_{\ell \in L}
    \chk {\with}\{\ell : A^\ell_m\}_{\ell \in L}
    \uses \lub_{\ell \in L} \Xi_{\ell}}
  {\Gamma \vdash e_\ell \chk A^\ell_m \uses \Xi_\ell\quad (\forall \ell \in L \neq \emptyset)}
  \\[1em]
  \infer[{\with}E]
  {\Gamma \vdash e.\ell \syn A_m^\ell \uses \Xi}
  {\Gamma \vdash e \syn {\with}\{\ell : A^\ell_m\}_{\ell \in L} \uses \Xi
    \quad (\ell \in L)}
  \\[1em]
  \infer[{\up}I]
  {\Gamma \vdash \mb{susp}\; e \chk \up_k^m A_k \uses \Xi\bbar_m}
  {\Gamma \vdash e \chk A_k \uses \Xi}
  \hspace{3em}
  \infer[{\up}E]
  {\Gamma \vdash \mb{force}\; s \syn A_k \uses \Xi}
  {\Gamma \vdash s \syn \up_k^m A_k \uses \Xi}
  \\[1em]
  \infer[{\tensor}I]
  {\Gamma \vdash (e_1,e_2) \chk A_m \tensor B_m \uses \Xi \join \Xi'}
  {\Gamma \vdash e_1 \chk A_m \uses \Xi
    & \Gamma \vdash e_2 \chk B_m \uses \Xi'}
  \\[1em]
  \infer[{\tensor}E]
  {\Gamma \vdash \mb{match}\; s\; ((x_1,x_2) \Rightarrow e') \chk C_r \uses \Xi \join (\Xi' \ctxrm x_1:A_m \ctxrm x_2:B_m)}
  {\Gamma \vdash s \syn A_m \tensor B_m \uses \Xi
    & m \geq r
    & \Gamma, x_1 : A_m, x_2 : B_m \vdash e' \chk C_r \uses \Xi'}
  \\[1em]
  \infer[{\one}I]
  {\Gamma \vdash (\,) \chk \one_m \uses (\cdot)}
  {}
  \hspace{3em}
  \infer[{\one}E]
  {\Gamma \vdash \mb{match}\; s\; ((\,) \Rightarrow e') \chk C_r \uses \Xi \join \Xi'}
  {\Gamma \vdash s \syn \one_m \uses \Xi
    & m \geq r
    & \Gamma \vdash e' \chk C_r \uses \Xi'}
  \\[1em]
  \infer[{\plus}I]
  {\Gamma \vdash \ell(e) \chk {\plus}\{\ell : A_m^\ell\}_{\ell \in L} \uses \Xi}
  {\Gamma \vdash e \chk A_m^\ell \uses \Xi \quad (\ell \in L)}
  \hspace{3em}
  \infer[{\plus}E_0]
  {\Gamma \vdash \mb{match}\; s\; (\,) \chk C_r \uses \Xi \semi [\Gamma|_r]}
  {\Gamma \vdash s \syn \oplus_m\{\,\} \uses \Xi
    & m \geq r}
  \\[1em]
  \infer[{\plus}E]
  {\Gamma \vdash \mb{match}\; s\; (\ell(x) \Rightarrow e_\ell)_{\ell \in L}
    \chk C_r \uses \Xi \join \lub_{\ell \in L}(\Xi'_\ell \ctxrm x:A_\ell)}
  {\Gamma \vdash s \syn {\plus}\{\ell : A_m^\ell\}_{\ell \in L} \uses \Xi
    & m \geq r
    & \Gamma, x : A_m^\ell \vdash e_\ell \chk C_r \uses \Xi'_\ell \quad (\forall \ell \in L \neq \emptyset)}
  \\[1em]
  \infer[{\down}I]
  {\Gamma \vdash \mb{down}\; e \chk \down^n_m A_n \uses \Xi}
  {\Gamma \vdash e \chk A_n \uses \Xi}
  \\[1em]
  \infer[{\down}E]
  {\Gamma \vdash \mb{match}\; s\; (\mb{down}\; x \Rightarrow e') \chk C_r \uses \Xi \join (\Xi' \ctxrm x:A_n)}
  {\Gamma \vdash s \syn \down^n_m A_n \uses \Xi
    & m \geq r
    & \Gamma, x : A_n \vdash e' \chk C_r \uses \Xi'}
\end{rules}
  \caption{Algorithmic Typing for Natural Deduction}
  \label{fig:algo}
\end{figure}

\begin{lemma}
  If $\Gamma \vdash e \Longleftrightarrow A_m \uses \Xi$
  then $\Xi \geq m$
\end{lemma}
\begin{proof}
  By straightforward rule induction.
\end{proof}

\begin{lemma}[Context Extension]
  If $\Gamma \vdash e \Longleftrightarrow A_m \uses \Xi$ and $\Xi \sqsupseteq \Delta$
  and $\Gamma' \supseteq \Gamma$ then $\Gamma' \vdash e \Longleftrightarrow A_m \uses \Xi'$
  for some $\Xi'$ with $\Xi' \sqsupseteq \Delta$.
\end{lemma}
\begin{proof}
  By rule induction on the given derivation and inversion on the definition of
  $\Xi \sqsupseteq \Delta$.
\end{proof}

\begin{proof}
  (of soundness, \autoref{thm:algo-sound}) By rule induction on the algorithmic
  typing derivation and inversion of the $\Xi \sqsupseteq \Delta$ judgment.
\end{proof}

\begin{proof}
  (of completeness, \autoref{thm:algo-complete}) By rule induction on the given
  bidirectional typing.
\end{proof}

\begin{corollary}
  $\cdot \vdash e \chk A_m$ iff $\cdot \vdash e \chk A_m / \cdot$
\end{corollary}

\section{Conclusion}
\label{sec:conclusion}

There have been recent proposals to extend the adjoint approach to combining
logics to dependent types.  \citet{Licata16lfcs,Licata17fscd} permit dependent
types and richer connections between the logics that are combined, but certain
properties such as independence are no longer fundamental, but have to be proved
in each case where they apply.  While they mostly stay within a sequent
calculus, they also briefly introduce natural deduction.  They also provide a
categorical semantics.  \citet{Hanukaev23tyde} also permit dependent types and
use the graded/algebraic approach to defining their system.  However, their
approach to dependency appears incompatible with control of contraction, so
their adjoint structure is not nearly as general as ours.  They also omit empty
internal choice (and external choice altogether), which created some of the
trickiest issues in our system.  Neither of these proposes an algorithm for type
checking or an operational semantics that would exploit the substructural and
mode properties to obtain ``free theorems'' about well-typed programs.

We are pursuing several avenues of future work building on the results of this
paper.  On the foundational side, we are looking for a direct algorithm to
convert an arbitrary natural deduction into a verification.  On the programming
side, we are considering \emph{mode polymorphism}, that is, type-checking the
same expression against multiple different modes to avoid code duplication.  On
the application side, we are considering staged computation, quotation, and
metaprogramming, decomposing the usual type $\Box A$ or its contextual analogue
along the lines of \autoref{ex:s4}.

\bibliographystyle{plainnat}
\bibliography{fp-copy,additional-bib}

\appendix

\end{document}